\newif\ifJFP
\JFPfalse
\ifJFP
\documentclass{jfp1}
\else
\documentclass[a4paper,12pt,fullpage]{article}
\fi
\ifJFP \else 
\usepackage{fullpage} 
\usepackage{amsmath}
\newenvironment{proof}[1]{\begin{quotation}\noindent\textsf{Proof:} #1}
{\(\Box\)\end{quotation}}
\usepackage{authblk}
\fi
\usepackage{marvosym}
\usepackage{alltt}
\usepackage{lscape}
\usepackage{url}
\usepackage[all]{xy}
\usepackage{qsymbols}
\usepackage{graphicx}
\usepackage{amssymb}
\usepackage{enumerate}

\newtheorem{theo}{Theorem}
\newtheorem{fact}[theo]{Fact}
\newtheorem{prop}[theo]{Proposition}
\newtheorem{lemma}[theo]{Lemma}
\newtheorem{corollary}[theo]{Corollary}

\newcommand{\ie}{i.e.,~}
\newcommand{\nat}{\ensuremath{\mathbb{N}}}

\newcommand{\T}{\mathcal{T}}
\newcommand{\F}{\mathcal{F}}
\newcommand{\G}{\mathcal{G}}

\newcommand{\Var}[1]{\underline{\mathsf{#1}}}
\newcommand{\PNF}{{}^{\tiny \textsf{NF}}\!P}
\newcommand{\QNF}{{}^{\tiny \textsf{NF}}\!Q}
\newcommand{\p}[2]{p^{[#1]}_{#2}}
\newcommand\subrel[2]{\mathrel{\mathop{#2}\limits_{#1}}}
\newcommand{\simtozero}{\subrel{z \to z_0}{\sim}}

\usepackage[usenames,dvipsnames]{color}
\definecolor{darkbrown}{cmyk}{.3,.75,.75,.15}
\definecolor{vertfonce}{rgb}{0,.5,0}

\newcommand{\bl}[1]{\textcolor{blue}{#1}}

\definecolor{vertfonce}{rgb}{0,.5,0}

\newif\ifcomment

\commenttrue           

\ifJFP
\title[Counting and generating lambda terms]{Counting and generating lambda terms}
\author[K. Grygiel and P. Lescanne]{Katarzyna Grygiel$^\dagger$
\thanks{This work was supported by the National Science Center of Poland, grant number 2011/01/B/HS1/00944, when the author hold a post-doc position at the Jagiellonian University within the SET project co-financed by the European Union.}\\
 \and\\
Pierre Lescanne$^{\dagger,\ddagger}$\\[10pt]
$^\dagger$Theoretical Computer Science Department, \\
Faculty of Mathematics and Computer Science,\\
Jagiellonian University, ul. Prof. {\L}ojasiewicza 6, 30-348 Krak\'ow, Poland\\[10pt]
ENS de Lyon, University of Lyon, \\
$^\ddagger$LIP (UMR 5668 CNRS ENS Lyon UCBL INRIA)\\ 46 all\'ee
d'Italie, 69364 Lyon, France\\
\email{grygiel@tcs.uj.edu.pl,pierre.lescanne@ens-lyon.fr}
}
\else
\title{Counting and generating lambda terms}
\author[1]{Katarzyna Grygiel\thanks{This work was supported by the National Science Center of Poland, grant number 2011/01/B/HS1/00944, when the author hold a post-doc position at the Jagiellonian University within the SET project co-financed by the European Union.}
\thanks{\textsf{email:} grygiel@tcs.uj.edu.pl}}
\author[1,2]{Pierre Lescanne\thanks{\textsf{email:} pierre.lescanne@ens-lyon.fr}
}
\affil[1]{Theoretical Computer Science Department

Faculty of Mathematics and Computer Science

Jagiellonian University 

ul. Prof. {\L}ojasiewicza 6, 30-348 Krak\'ow, Poland
\medskip
}

\affil[2]{ENS de Lyon

 LIP  (UMR 5668 CNRS ENS Lyon UCBL INRIA)

University of Lyon

46 all\'ee d'Italie, 69364 Lyon, France
}
\fi

\begin{document}

\maketitle

\begin{abstract}
  Lambda calculus is the basis of functional programming and higher order proof
  assistants.  However, little is known about combinatorial properties of lambda
  terms, in particular, about their asymptotic distribution and random
  generation. This paper tries to answer questions like: How many terms of a given
  size are there?  What is a ``typical'' structure of a simply typable term?  Despite
  their ostensible simplicity, these questions still remain unanswered, whereas
  solutions to such problems are essential for testing compilers and optimizing
  programs whose expected efficiency depends on the size of terms.  Our approach
  toward the aforementioned problems may be later extended to any language with bound
  variables, i.e., with scopes and declarations.

  This paper presents two complementary approaches: one, theoretical, uses complex
  analysis and generating functions, the other, experimental, is based on a generator
  of lambda terms.  Thanks to de Bruijn indices, we provide three families of
  formulas for the number of closed lambda terms of a given size and we give four
  relations between these numbers which have interesting combinatorial
  interpretations.  As a by-product of the counting formulas, we design an algorithm
  for generating $`l$-terms.  Performed tests provide us with experimental data, like
  the average depth of bound variables and the average number of head lambdas. We
  also create random generators for various sorts of terms.  Thereafter, we conduct
  experiments that answer questions like: What is the ratio of simply typable terms
  among all terms?  \emph{(Very small!)}  How are simply typable lambda terms
  distributed among all lambda terms?  \emph{(A~typable term almost always starts
    with an abstraction.)}

  In this paper, abstractions and applications have size $1$ and variables have
  size~$0$.

  \medskip

  \noindent \textbf{Keywords:} lambda calculus, combinatorics, functional
  programming, test, random generator, ranking, unranking

\end{abstract}

\section{Introduction}
\label{sec:introduction}

Let us start with a few questions relevant to the problems we address.
\begin{itemize}
\item How many closed $`l$-terms are of size $50$ (up to $`a$-conversion)?
\[996657783344523283417055002040148075226700996391558695269946852267.\]
\item How many closed terms of size $n$ are there?\\
\emph{We will give a recursive formula for this number in Section~$\ref{sec:polyn-count-lambd}$.}
\item What does the following sequence enumerate:
  \[ 0, 1, 3, 14, 82, 579, 4741, 43977, 454283, 5159441, 63782411?\] \emph{This
    sequence enumerates closed terms of size $0$, $1$, $2$, $3$, $4$, $5$, $6$, $7$,
    $8$, $9$, $10$.  It is the sequence A220894 of the Online Encyclopedia of Integer
    Sequences (\url{https://oeis.org/A220894}).  We will provide three ways to compute
    it (Section~$\ref{sec:three-formulas}$). }
\item Is it possible to generate simply typable terms randomly?\\
  \emph{Yes, according to the process which consists in generating random $`l$-terms
    with uniform probability and sieving those that are simply typable.  Thus, we can
    generate random simply typable terms of size up to $50$.}
\item Is a term starting with an abstraction more likely to be typable than a term
  starting with an application?\\
\emph{The answer is positive as shown in Figure~$\ref{fig:dist_typed_size_25}$,
  which gives the distribution of simply typable $`l$-terms among all $`l$-terms.}
\item Do these results have practical consequences?\\
\emph{Yes, they enable random generation of simply typable terms in the case of terms of size
  up to $50$ in order to debug compilers or other programs,
  manipulating terms, e.g., type checkers or pretty printers.}
\end{itemize}

The above questions seem rather classical, but amazingly very little is known about
combinatorial aspects of $`l$-terms, probably because of the intrinsic difficulty of
the combinatorial structure of lambda calculus due to the presence of bound
variables.    However, the answers to these questions are extremely
important not only for a better understanding of the structure of $`l$-terms, but also
for people who build test samples for debugging compilers.  Perhaps the
reason for this ignorance lies in the surprising form of the recurrences.  Indeed, due to the
presence of bound variables, the recurrence does not work in the way mathematicians expect
and are used to.  The induction lies on two variables, one decreasing (the size),
the other increasing (the number of free variables). 
Thus none of the methods used in the reference book of
Flajolet and Sedgewick \cite{flajolet08:_analy_combin} applies.  Why is that?  In what
follows we compute the number of $`l$-terms (and of normal forms) of size $n$ with at
most $m$ distinct free variables.  Denoting the number of such terms by $T_{n,m}$, the formula for
$T_{n,m}$ contains $T_{n-1,m+1}$ and this growth of $m$ makes the formula averse to
treatments by generating functions and classical analytic combinatorics.  We notice that
for a given $n$ the expression for $T_{n,m}$ is a polynomial in $m$.  These polynomials can
be described inductively and their coefficients are given by recurrence formulas.  These formulas
are still complex, but can be used to compute the constant coefficients, which correspond
to the numbers of closed $`l$-terms.  For instance, the leading coefficients of the
polynomials are the well known Catalan numbers which count binary trees.

In order to find the recurrence formula for the number of \mbox{$`l$-terms} of a
given size, we make use of the representation of variables in $`l$-terms by de Bruijn
indices.  Recall that a de Bruijn index is a natural number which replaces a term
variable and enumerates the number of $`l$'s encountered on the way between the
variable and the $`l$ which binds the latter.  In this paper, we assume the
combinatorial model in which the size of each occurrence of abstraction or
application is counted as $1$, while the size of variables (de~Bruijn indices)
as~$0$.  This method is a realistic model of the complexity of $`l$-terms and allows
us to derive the recurrences very naturally.  

From the formula for counting $`l$-terms we derive one-to-one assignments of terms of
size $n$ with at most $m$ distinct free indices to the numbers in the interval
$[1..T_{n,m}]$.  From this correspondence, we develop a program for generating
$`l$-terms, more precisely for building $`l$-terms associated with numbers in the
interval $[1..T_{n,m}]$.  In combinatorics the function that counts objects by
assigning a number to each object is called a \emph{ranking} and its inverse, i.e.,
the function that assigns an object to a rank is called an
\emph{unranking}~\cite{integer:ranking}.  Thus, in this paper, we can say that we rank and unrank
lambda-terms and normal forms.  If we pick a random number in the interval
$[1..T_{n,m}]$, then we get a random term of size $n$ with at most $m$ distinct free
variables.  Most of the time we consider closed $`l$-terms, which means $m=0$.
Beside the interest in such a random generation for applications like testing, this
allows us to compute practical values of parameters by Monte-Carlo methods.  Overall,
we are able to build a random generator for simply typable terms. Unlike the method
used so far \cite{Palka:2011:TOC:1982595.1982615}, which consists in unfolding the
typing tree, we generate random \mbox{$`l$-terms} and test their typability, until we
find a simply typable term.  This method allows us to generate uniformly simply
typable $`l$-terms up to size $50$.  We also use this method to describe the
distribution of typable terms among all terms and typable normal forms among all
normal forms.

\subsection*{Structure of the paper}
\label{sec:structure-ther-paper}

According to its title, the paper is divided into two parts, the first one focuses on counting terms
and its mathematical treatment, the second one addresses term generation and its applications. The
first part (Sections~\ref{sec:polyn-count-lambd} and \ref{sec:generating}) is devoted to
the formulas counting $`l$-terms.  In Section~\ref{sec:polyn-count-lambd} we study
polynomials giving the numbers of terms of size~$n$ with at most $m$ distinct free
variables.  In
Section~\ref{sec:an-comb-interpr}, we show that the numbers of $i$-contexts give a
combinatorial interpretation of the coefficients of the polynomials and yield a new formula
for counting the closed terms of size $n$. If we add formulas for counting $`l$-terms of
size $n$ with exactly $m$ distinct free variables, we have three formulas of three different origins for
counting closed terms which we describe and compare in Section~\ref{sec:three-formulas}.
In Section~\ref{sec:generating} we derive generating functions and asymptotic values for
these coefficients.  In Section~\ref{sec:count-norm-forms} we give a formula for counting
normal forms.  In the second part of the paper, \ie in Section~\ref{sec:lambda-term-gener} and
Section~\ref{sec:simply-typed-terms}, we propose programs to generate untyped and typable
terms and normal forms.  Section~\ref{sec:exper-data} is devoted to experimental
results.  Section~\ref{sec:related-work} presents related works.




\section{Counting terms with at most $m$ distinct free variables}
\label{sec:polyn-count-lambd}

We represent terms using de Bruijn indices \cite{deBruijn72}, which means that
variables are represented by numbers $\Var{1}, \Var{2}, \ldots, \Var{m}, \ldots$,
where an index, for instance~$\Var{k}$, is the number of $`l$'s above the location of
the index and below the~$`l$ that binds the variable, in a representation of
$`l$-terms by trees.  For instance, the term with variables $`l x. `l y. x\,y$ is
represented by the term with de Bruijn indices $`l `l \Var{2} \Var{1}$. The variable
$x$ is bound by the top~$`l$.  Above the occurrence of~$x$ there are two $`l$'s,
therefore $x$ is represented by $\Var{2}$, and from the occurrence of~$y$ we count
just the~$`l$ that binds $y$, so $y$ is represented by $\Var{1}$.  Notice that unlike
\ifJFP\else Lescanne~\cite{LescannePOPL94} \fi and like
\ifJFP
\else de Bruijn~\cite{deBruijn72}  and Abadi \emph{et
  al.}~\cite{AbadiCCL91JFP} \fi we start
indices at $1$, since it fits better with our aim of counting terms.

In what follows, by \emph{terms} we mean untyped terms with de Bruijn indices and we
often speak indistinctively of variables and (de Bruijn) indices.  Assume that in a
term~$t$ not all occurrences of indices need to be bound, i.e., there may occur indices that do
not correspond to surrounding $`l$'s. Such indices are called ``free'' in $t$.  Now,
we introduce the notational convention for ``free'' indices occurring in terms. An
\emph{interval of free indices} for a term $t$ is a set $\{\Var{1}, \Var{2}, \ldots,
\Var{m}\}$ of indices, written $[\Var{1}..\Var{m}]$, such that
\begin{enumerate}[(i)]
\item if $t$ is an index $\Var{i}$, then any interval $[\Var{1}..\Var{m}]$ with $1\le i\le m$ is an interval for $t$,
\item if $t$ is an abstraction $`l s$ and an interval of free indices for $s$ is
  $[\Var{1} .. \Var{m+1}]$, then the interval of free indices for $t$ is
  $[\Var{1} .. \Var{m}]$ (since the index $\Var{1}$ is now bound and the others are
  assumed to decrease by one),
\item if $t$ is an application $t_1 t_2$ and an interval of indices for $t_1$ and
  $t_2$ is $[\Var{1} .. \Var{m}]$, then an interval of indices for $t$ is
  $[\Var{1} .. \Var{m}]$.
\end{enumerate}
To illustrate (ii), assume $t = `l s = `l \Var{3}\,\Var{1}$. An interval of free
indices for $s$ is $[\Var{1} .. \Var{m+1}]$ for any $m\ge 2$.  For
instance for $m=3$, $[ \Var{1},\Var{2}, \Var{3}, \Var{4} ]$ is an interval of free
indices for $s$.  For ${m=2}$, $[ \Var{1},\Var{2}, \Var{3} ]$ is another interval of
free indices for $s$. An interval of free indices for $t$ is $[\Var{1} .. \Var{m}]$ 
for any $m\ge 2$ and for $m=3$, $[ \Var{1},\Var{2}, \Var{3}]$ is an
interval of free indices for $t$.  For $m=2$, $[\Var{1},\Var{2}]$ is another interval
of free indices for $t$.  To say it in rough words, whereas one sees $\Var{3}$ as
$\Var{3}$ in $s$, one sees $\Var{3}$ as $\Var{2}$ in $t$ due to the abstraction $`l$
which decreases the indices as they are seen.

We measure the size of a term in the following way:
\begin{eqnarray*}
|\Var{m}| &=& 0, \textrm{~for every index~} \Var{m},\\
|`l t| &=& |t| + 1,\\
|t s| &=& |t| + |s| + 1.
\end{eqnarray*}

Since $`l$-terms can be represented as unary-binary trees with labels or pointers,
the notion of size of a term $t$ corresponds to the number of unary and binary
vertices in the tree representing $t$. This also means that adding a new variable (in
other words, adding a new leaf to a tree) or a new operator (a unary or a binary
vertex) always increases the size of a term by $1$.

One can define $m$ using the concept of term openness (due to John Tromp).  The \emph{openness} of a
terms is the minimum number of outer $`l$'s necessary to close the terms, i.e., to
make the term a closed term.  For instance, the openness of $(`l x. (x\,y))(`l x. (x\,z))$ is equal to $2$ since the term needs two abstractions to become closed.

Let us denote by $\T_{n,m}$ the set of terms of size $n$ with at most $m$ distinct
free de Bruijn indices.  $\T_{n,m}$ is isomorphic to the set of terms having an
openness equal to at most~$m$.  In what follows, we use the symbol $@$ to denote
applications, whereas classical theory of $`l$-calculus uses concatenation, which we
find not explicit enough for our purpose.

\begin{eqnarray*}
\T_{0,m} &=& [\Var{1} .. \Var{m}] \\
\T_{n+1,m} &=& `l \T_{n,m+1} \ \uplus\ \biguplus_{i=0}^n\T_{i,m} @ \T_{n-i,m}.
\end{eqnarray*}

For all $n,m \in {\mathbb N}$, let $T_{n,m}$ denote the cardinality of the set
$\T_{n,m}$. According to the
definition of size, operators $`l$ and $@$ have size $1$ and de Bruijn indices have
size~$0$. Therefore, we get the following two equations specifying $T_{n,m}$:

\begin{eqnarray*}
 T_{0,m} &=& m\\
 T_{n+1,m} & = &T_{n,m+1} + \sum_{i=0}^{n} T_{i,m} T_{n-i,m}.
\end{eqnarray*}

This means that there are $m$ terms of size $0$ with at most $m$ distinct free de Bruijn indices, which are terms
that are just these indices.   Terms of size $n+1$ with at most $m$ distinct free de
Bruijn indices are either abstractions with at most $m+1$ distinct free indices on a term of size
$n$ or applications of terms with at most $m$ distinct free indices to make a term of size $n+1$.  As we said in the introduction, the $11$ first values of $T_{n,0}$ are:
\[0, 1, 3, 14, 82, 579, 4741, 43977, 454283, 5159441, 63782411.\] %
$T_{n,0}$ is sequence \textbf{A220894} in the \emph{On-line Encyclopedia of Integer Sequences}.
\begin{figure*}
  \centering
  \begin{tiny}
    \begin{math}
      \begin{array}[c]{ c r  r  r  r  r  r  r }
        \mathit{\mathbf{n\backslash m}} & \mathbf{0} & \mathbf{1} & \mathbf{2} & \mathbf{3} & \mathbf{4} & \mathbf{5} & 6  \\
        \hline 
        \mathbf{0} &  0 & 1 & 2 & 3 & 4 & 5 & 6   \\
        \mathbf{1}&1 & 3 & 7 & 13 & 21 & 31 & 43  \\
        \mathbf{2}&3 & 13 & 41 & 99 & 199 & 353 & 573 \\
        \mathbf{3}&14 & 76 & 312 & 962 & 2386 & 5064 & 9596  \\
        \mathbf{4}&82 & 542 & 2784 & 10732 & 32510 & 82122 & 181132 \\
        \mathbf{5}&579 & 4493 & 27917 & 131715 & 482015 & 1440929 & 3687513 \\
        \mathbf{6}&4741 & 42131 & 307943 & 1741813 & 7612097 & 26763551 & 79193491 \\
        \mathbf{7}&43977 & 439031 & 3690055 & 24537945 & 126536933 & 519788827 & 1771730211 \\
        \mathbf{8}&454283 & 5020105 & 47635777 & 365779679 & 2198772055 & 10477986133 &
        40973739725 \\
        \mathbf{9}&5159441 & 62382279 & 658405747 & 5744911157 & 39769404045 & 218213327131 &
        974668783199 \\
        \mathbf{10}&63782411 & 835980065 & 9695617821 & 94786034723 & 746744227319 &
        4681133293821 & 23769847893305 \\
        \mathbf{11}&851368766 & 12004984120 & 151488900012 & 1639198623818 & 14531624611594 &
        103244315616876 & 593009444765240 \\
       \mathbf{12} &12188927818 & 183754242626 & 2502346785164 & 29658034018852 &
        292747054367966 & 2338363467319958 & 15112319033576416 \\
        \mathbf{13} &186132043831 & 2984264710781 & 43560247035581 & 560484305049943 &
        6100545513799835 & 54347237563601321 & 393031286917940401
        \\
        \mathbf{14} &3017325884473 & 51220227153987 & 796828655891895 & 11046637024014049 &
        131425939696979805 & 1295642289776992983 & 10425601907159190187\\
       \end{array}
    \end{math}
  \end{tiny}
  \caption{Values of $T_{n,m}$ for $n$ and $m$ up to $14$ and $6$, respectively}
  \label{fig:T_n_m}
\end{figure*}

Figure~\ref{fig:T_n_m} gives all the values of $T_{n,m}$ for $n$ up to $14$ and $m$ up
to~$6$.  For instance, there is $1$ closed term of size $1$, namely $`l\Var{1}$, there are
$3$ closed terms of size $2$, namely $ `l `l \Var{1}, `l `l \Var{2}, `l \Var{1}\, \Var{1}$,
and there are $14$ closed terms of size $3$, namely


\begin{center}
  \begin{math}
    \begin{array}{ccccccc}
`l`l`l \Var{1} & `l`l`l \Var{2} & `l`l`l \Var{3} & `l`l \Var{1}\,\Var{1} & `l`l \Var{1}\,\Var{2}& 
      `l`l \Var{2}\,\Var{1} & `l`l \Var{2}\,\Var{2} \\ 
`l (\Var{1}\, `l \Var{1}) & `l (\Var{1}\, `l \Var{2}) & `l \Var{1} (\Var{1} \Var{1}) & `l ((`l\Var{1})\, \Var{1} ) & `l ((`l\Var{2})\, \Var{1} ) & `l((\Var{1}\,\Var{1}) \, \Var{1}) & (`l \Var{1})\, `l \Var{1}.
  \end{array}
  \end{math}
\end{center}

Notice that in Section~\ref{sec:lambda-term-gener} we describe how to assign a term 
to a number and therefore how to list terms with increasing numbers.  The above terms
are listed in that order.

\subsection{Computing the $T_{n,m}$'s}
\label{sec:actual}

The recursive definition of $T$ yields an easy naive program in a functional
programming language (here \textsf{Haskell}):
\begin{verbatim}
naiveT :: Int -> Int -> Integer
naiveT 0 m = fromIntegral m
naiveT n m = naiveT (n-1) (m+1) + 
             sum [naiveT i m * naiveT (n-1-i) m | i <- [0..n-1]]
\end{verbatim}
This program is inefficient since it recomputes the values of $T$ at each recursive
call.   For actual computations a program with memoization is required.  In
\textsf{Sage}  this is obtained by requiring the function to be ``cached''. In
\textsf{Haskell} we use the laziness of streams:
\begin{verbatim}
ttab' :: [[Integer]]
ttab' = [0..] : [[t' (n-1) (m+1) + s n m | m <- [0..]] | n <- [1..]]
  where s n m = sum $ zipWith (*) (ti n m) (reverse $  ti n m)
        ti n m = [t' i m | i <- [0..(n-1)]]

t' :: Int -> Int -> Integer
t' n m = ttab !! n !! m
\end{verbatim}
This program is not efficient enough and John Tromp proposed us a better program:
\begin{verbatim}
ttab :: [[[Integer]]]
ttab = iterate nextn . map return $ [0..] 
  where
  nextn ls = zipWith rake (tail ls) ls
  rake (m1:_) ms = (m1 + conv ms) : ms
  conv ms = sum $ zipWith (*) ms (reverse ms)
  
t :: Int -> Int -> Integer
t n m = head $ ttab !! n !! m
\end{verbatim}
Assume that we compute \textsf{ttab~n~0} for the first time.  The basic operation
\textsf{rake} requires $O(n)$ additions and $O(n)$ multiplications. \textsf{nextn}
requires $O(n)$ calls to \textsf{rake} and \textsf{iterate} requires  $O(n)$ to \textsf{nextn}.
Therefore, the complexity
of the first computation of \textsf{t n 0} depends on the complexity of the addition
and of the multiplication of arbitrary-precision integers which we may assume (intuitively) to be
$O(log^2(T_{n,0}))$. Although the question of the asymptotic size of the number
$T_{n,0}$ is open, we know that it is superexponential in $n$ and, on the other hand,
it is asymptotically smaller than~$n^n$. Therefore, \texttt{t n 0} runs in
$O(n^3\times log^2(T_{n,0}))$ which is at least of order~$n^5$ and at most of
order~$n^5 \log^2 (n)$.  Such estimations seem to be in accordance with our
experiments.  Once the table \texttt{ttab} is constructed, the runtime of \texttt{t}
is in $O(n+m)$.

\subsection{The polynomials $P_n$}
\begin{it}
  \begin{small}
    In the end of this section and in the four coming sections we present results of
    mainly combinatorial flavor. We focus there on the quantitative approach to
    lambda calculus, with special emphasis on the challenging problem of counting
    $`l$-terms and approximating its asymptotic behavior.  Therefore, a reader
    interested mostly in term generation and experimental results can skip this
    material and go directly to Section \ref{sec:lambda-term-gener}.
  \end{small}
\end{it}

\medskip

The problem of determining the asymptotic estimation of the number of closed terms of
a given size turns out to be a non-trivial task. Due to the unusual combinatorial
structure of $`l$-terms, such objects seem to resist methods developed in
combinatorics so far. There are a few papers devoted to this challenging problem
\cite{gittenberger-2011-ltbuh,
  DBLP:journals/corr/abs-0903-5505,DBLP:journals/tcs/Lescanne13}, however, none of
the methods used by now could provide the final solution.  \ifJFP
\else Bodini \emph{et
  al.}~\cite{gittenberger-2011-ltbuh} \fi use essentially analytic methods exploiting
the functional equation of Proposition~\ref{prop:genfun}
(Section~\ref{sec:generating-function}), whereas \ifJFP
\else  David \emph{et al.}
\cite{DBLP:journals/corr/abs-0903-5505} \fi use upper and lower bound approximations
and \ifJFP\else
Lescanne~\cite{DBLP:journals/tcs/Lescanne13} \fi uses algebraic computations on
polynomials and power series.  Our approach can be considered as the development of
the previous research carried out by \ifJFP
\else \cite{DBLP:journals/tcs/Lescanne13}\fi.

For every $n \geq 0$, we associate with $T_{n,m}$ a polynomial $P_n(m)$ in $m$. First, let us define polynomials $P_n$ in the following recursive way:
\begin{eqnarray*}
 P_0(m) &=& m,\\
 P_{n+1}(m) &= & P_n(m+1) + \sum_{i=0}^{n} P_{i}(m) P_{n-i}(m).
\end{eqnarray*}

The sequence $\left( P_n(0) \right)_{n \geq 0}$ corresponds to the sequence $\left(
  T_{n,0} \right)_{n \geq 0}$ enumerating closed \mbox{$`l$-terms}.  The first nine
polynomials are given in Figure~\ref{fig:poly}.

\begin{figure*}
  \centering
  \begin{scriptsize}
  \begin{math}
    \begin{array}[c]{ c  l }
      \mathbf{n} &\mathit{\hspace{150pt}P_n}\\
\hline
      \mathbf{0} & m \\
      \mathbf{1} & m^2 + m + 1\\
      \mathbf{2} & 2 m^3 + 3 m^2 + 5 m + 3\\
\mathbf{3} & 5 m^4 + 10 m^3 + 22 m^2 + 25 m + 14\\
\mathbf{4} & 14 m^5 + 35 m^4 + 94 m^3 + 154 m^2 + 163 m + 82\\
\mathbf{5} & 42 m^6 + 126 m^5 + 396 m^4 + 838 m^3 + 1277 m^2 + 1235 m + 579\\
\mathbf{6} & 132 m^7 + 462 m^6 + 1654 m^5 + 4260 m^4 + 8384 m^3 + 11791 m^2 + 10707 m
+ 4741\\
\mathbf{7} & 429 m^8 + 1716 m^7 + 6868 m^6 + 20742 m^5 + 49720 m^4 + 90896 m^3 +
120628 m^2 + 104055 m + 43977\\
\mathbf{8} & 1430 m^9 + 6435 m^8 + 28396 m^7 + 98028 m^6 + 275886 m^5 + 617096 m^4 +
1068328 m^3 + 1352268 m^2 + 1117955 m + 454283\\
    \end{array}
  \end{math}
  \end{scriptsize}
  \caption{The first nine polynomials $P_n$}
  \label{fig:poly}
\end{figure*}

This means that the constant coefficient of a polynomial $P_n(m)$ is exactly the number of
closed $`l$-terms of size $n$.  

\begin{lemma}\label{lem:degree}
For every $n$, the degree of the polynomial $P_n$ is equal to $n+1$.
\end{lemma}

\begin{proof}
The result follows immediately by induction on $n$ from the definition of $P_n$.
\end{proof}

For $i>0$ and $n\geq 0$, let us denote by $\p{i}{n}$ the $i^{th}$ leading coefficient of the polynomial $P_n$, \ie we have
\[P_n(m) = \p{1}{n} m^{n+1} + \p{2}{n} m^n + \ldots + \p{i}{n} m^{n+2-i} + \ldots + \p{n+1}{n} m + \p{n+2}{n}.\]

\begin{lemma}\label{p-values}
For every $n \geq 0$ and $i > 0$,
\begin{eqnarray*}
\p{1}{0} &=& 1, \quad \p{i}{0} = 0 \ \textit{~for~} \ i>1,\\
\p{i}{n+1} &= & \sum_{j=0}^{i-2} {n+1-j \choose i-2-j} \p{j+1}{n} + \sum_{k=1}^i \sum_{j=0}^{n} \p{k}{j} \p{i+1-k}{n-j}.
\end{eqnarray*}
\end{lemma}

\begin{proof}
Since $P_0(m)=m$, equations from the first line in the above lemma are trivial.

The $i^{th}$ leading coefficient in the polynomial $P_{n+1}(m)$ is equal to the sum of coefficients standing at $m^{n+3-i}$ in polynomials $P_n(m+1)$ and $\sum_{j=0}^{n} P_j(m) P_{n-j}(m)$.

The first of these polynomials, $P_n(m+1)$, is as follows:
\[\p{1}{n} (m+1)^{n+1} + \ldots + \p{i-1}{n} (m+1)^{n+3-i} + \ldots + \p{n+2}{n},\]
therefore the coefficient of $m^{n+3-i}$ in $P_n(m+1)$ is equal to
  \[{n+1 \choose i-2}\p{1}{n} + {n \choose i-3}\p{2}{n} + \ldots + {n+3-i \choose 0}\p{i-1}{n} \
  = \ \sum_{j=0}^{i-2} {n+1-j \choose i-2-j} \p{j+1}{n}.\]
In the case of the second polynomial, $\sum_{j=0}^{n} P_j(m) P_{n-j}(m)$, we have
\begin{eqnarray*}
\left( \p{1}{j} m^{j+1} + \ldots + \p{k}{j} m^{j+2-k} + \ldots + \p{j+2}{j} \right) \\
\cdot \left( \p{1}{n-j} m^{n-j+1} + \ldots + \p{i+1-k}{n-j} m^{n-j+1+k-i} + \ldots + \p{n-j+2}{n-j} \right),
\end{eqnarray*}
therefore the coefficient of $m^{n+3-i}$ in $\sum_{j=0}^{n} P_j(m) P_{n-j}(m)$ is equal to
\[\sum_{k=1}^i \sum_{j=0}^{n} \p{k}{j} \p{i+1-k}{n-j}. \]
\end{proof}
The next section proposes a combinatorial interpretation of the coefficients $\p{i}{j}$.

\section{Counting contexts}
\label{sec:an-comb-interpr}

In $`l$-calculus, an $i$-context is a closed term with $i$ holes.  Variables and
holes are similar in the sense that they can be replaced by terms.  But whereas a
variable may occur many times in a term and so may be replaced by terms at more than
one place at a time, a hole is anonymous, occurs once and only once (like a linear
variable in linear $`l$-calculus~\cite{DBLP:journals/corr/abs-cs-0501035}) and can be
filled only once.  Since as we said holes look like anonymous variables occurring
once we suppose that each hole has size~$0$ and we assume that the holes are numbered
$1, \ldots, i$ as they appear in the term from left to right.  For instance, if we
denote every hole by $[\ ]$, then $(`l \Var{1} [\ ])`l `l [\ ]\Var{2}$ is a
$2$-context of size $6$ and its holes are numbered as follows $(`l \Var{1} [\
]_{\bl{1}})`l `l [\ ]_{\bl{2}}\Var{2}$. $0$-contexts correspond to closed terms.
There is only one $1$-context of size~$0$ and there are no $i$-contexts of size~$0$
for $i\neq 1$.  Let us write $c_{n,i}$ for the number of $i$-contexts of size $n$.
Then we have
\begin{displaymath}
\left.\begin{array}{rcl}
c_{0,1} &=& 1\\
c_{0,i} &=& 0 \textrm{~for~} i \neq 1.
\end{array}\right\}(\dagger)
\end{displaymath}

Now, let us see how we construct an $i$-context of size $n+1$ from smaller ones.
\begin{description}
\item[By abstraction:] let us take a $j$-context (for $j \in [i..n+1]$) of size $n$ and add a new lambda above it. Then we choose a set of $j-i$ holes among the $j$ holes which we substitute by variables (or indices) abstracted by the new lambda.  For a fixed $j$ there are ${j \choose i}\,c_{n,j}$
  such $i$-contexts.  Finally, we sum these quantities over every $j$ from $i$ to $n+1$ to get the numbers
  of $i$-contexts constructed this way.
\item[By application:] let us apply a $j$-context of size $k$ to an $(i-j)$-context of size $n-k$ (for $j \in [0..i]$ and $k \in [0..n]$). This gives us an $i$-context of size $n+1$ since the application operator has size $1$. For fixed $j$ and $k$ there are $c_{k,j}c_{n-k,i-j}$ such $i$-contexts. Finally, we sum these numbers from $j=0$ to $j=i$ and from $k=0$ to $k=n$.
\end{description}
Hence, we get the following formula:
\[c_{n+1,i} \quad = \quad \sum_{j=i}^{n+1} {j \choose i}\,c_{n,j} + \sum_{j=0}^i
\sum_{k=0}^nc_{k,j} c_{n-k,i-j}. \qquad (\star)\] %
Let us see how we can build terms from contexts.  Recall that, by construction, an
$i$-context has only holes and no free index, which means that all the indices are
bound.  Therefore to build a term of size $n$ with $i$ occurrences of free indices
taken among $m$ ones from an $i$-context of size $n$ and a map $f$ from $[1..i]$ to
$[1..m]$, we insert the index $f(j)$ in the $j^{th}$ hole.  There are $c_{n,i} m^i$
such terms.  Therefore
\[T_{n,m} \quad = \quad c_{n,n+1} m^{n+1} + \ldots + c_{n,i} m^i + \ldots + c_{n,0}\] is the
number of $\lambda$-terms of size $n$ with at most $m$ distinct free variables, which is the
polynomial~$P_n(m)$.   In particular, $c_{n,n+2-i} = \p{i}{n}$.   This can be written as follows:
\begin{displaymath}
  P_n(m) = \sum_{i=0}^{n+1} c_{n,i} m^i.
\end{displaymath}
The coefficients $c_{n,i}$ of the polynomials $P_n$'s count the $i$-contexts
of size $n$.  We see that $c_{n,i}=0$ when $i>n+1$.

\paragraph{The case $i=n+2$.}
\label{sec:case-i=n+1}

In the case when $i=n+2$, using the fact that $c_{n,i}=0$ for $i>n+1$,  the
equations $(\dagger)$ and $(\star)$ boil down to:
\begin{eqnarray*}
c_{0,1} &=& 1\\
c_{n+1,n+2} & = & \sum_{k=0}^n c_{k,k+1} c_{n-k,n-k+1},
\end{eqnarray*}
which is characteristic of the Catalan numbers.   Indeed, $(n+1)$-contexts of size $n$
have only applications and no abstractions and are therefore binary trees.

\subsection{The generating function for $(c_{n,i})_{n,i \in \nat}$}
\label{sec:generating-function}

\begin{prop}\label{prop:genfun}
  Consider the bivariate generating function 
  \( \displaystyle L(z,u) \ = \sum_{n,i\ge 0} c_{n,i} z^n u^i.\) Then \[{L(z,u)=u + z
    L(z,u+1) + z L(z,u)^2}.\]
\end{prop}
\begin{proof}{}
  Notice that
  \begin{eqnarray*}
    L(z,u) &=& \sum_{n =0}^{\infty} \left(\sum_{i=0}^{\infty}c_{n,i}  u^i\right)z^n\\
    &=& \sum_{n =0}^{\infty} P_n(u) z^n\\
    &=& u + z \sum_{n =0}^{\infty} P_{n+1}(u)z^n\\
    &=& u + z\sum_{n =0}^{\infty} P_n(u+1) z^n+  z \sum_{n =0}^{\infty} \sum_{k=0}^n P_{k}(u)
    P_{n-k}(u) z^n\\
    &=& u + z L(z,u+1) + z L(z,u)^2.
  \end{eqnarray*}
\end{proof}
A similar equation was known from Bodini, Gardy and Gittenberger
\cite{gittenberger-2011-ltbuh} (for variable size~$1$).  However, notice
that what they call $\mathcal{L}$ is not the class of open \mbox{$`l$-terms,} but the class
of $i$-contexts.  Notice that the function $L(z,0)$ is the generating function for the
number of closed terms of size $n$. 

The equation
\begin{displaymath}
  z L(z,u)^2 - L(z,u)  + u + z L(z,u+1) = 0
\end{displaymath}
has the following solution
\begin{eqnarray*}
  L(z,u) &=& \frac{1- \sqrt{1-4z (u + z L(z,u+1))}}{2z}.
\end{eqnarray*}
 Let us state
\[M(z,u) = 2z \, L(z,u) .\]
Then
\[M(z,u) = 1- \sqrt{1-4 z u - 2z M(z,u+1)}\]
and hence
\begin{footnotesize}
  \[\hspace{-.4cm}M(z,0) = 1- \sqrt{1 - 2z(1- \sqrt{1-4z - 2z(1- \sqrt{1-8z - 2z(1-
        \sqrt{1-12z - 2z(1-
        \sqrt{1-16z - \ldots })})})})}.\]
\end{footnotesize}

\subsection{The asymptotic behavior of $T_{n,0}$}
\label{sec:two}

The function $M(z,u)$ has a singularity $z_u$ for $1-4z_u u - 2z_u M(z_u,u+1) = 0$, in addition to
the singularities of $M(z,u+1)$.  Notice that since $z_u M(z_u,u+1)>0$, we get $z_u<\frac{1}{4u}$.  Therefore
$L(z,0)$ has a sequence of singularities $(z_u)_{u`:\nat}$ which tends to $0$.  Thus
the radius of convergence of $L(z,0)$ is 0.  Recall that a fundamental theorem on analytic
functions connects the radius of convergence of a generating function with the
exponential growth of its coefficients (see Section IV.3 \emph{Singularities and
  exponential growth of the coefficients} in \ifJFP
\else{Flajolet and Sedgewick's book
  \cite{flajolet08:_analy_combin}}\fi).  This theorem says that if a generating
function has a radius of convergence $R$, then its coefficients grow like
$\left(\frac{1}{R}\right)^n$.  This means that if the radius of convergence is $0$,
then the coefficients grow faster that $a^n$ for any $a`:\ensuremath{\mathbb{R}}$.
Such a behavior is called superexponential.

\section{Three formulas for counting closed terms}
\label{sec:three-formulas}

We have found three formulas to compute the number of closed terms of size $n$.  Let
us summarize them.  In what follows the bracketed notation $[k=j]$ is the function
which is~$1$ if $k=j$ and $0$ if $k\neq j$.

\paragraph{Case $m=0$ for terms with at most $m$ distinct free variables}~
\label{sec:case_atmost}

$T_{n,0}$ where
\begin{eqnarray*}
 T_{0,m} &=& m\\
 T_{n+1,m} & =& T_{n,m+1} + \sum_{i=0}^{n} T_{i,m} T_{n-i,m}.
\end{eqnarray*}
This formula is clearly the simplest.  Its simplicity, one sum and no binomial,
allows it to be unfolded and used as a basis for programming a term generator (see
Section~\ref{sec:lambda-term-gener}).

\paragraph{Case $m=0$ for terms with exactly $m$ distinct free variables}~
\label{sec:case_exactly}

$f_{n,0}$ where
\begin{eqnarray*}
f_{0,m} & =& [m = 1] \\
f_{n,m} &=& 0 \textrm{~ if}~ m>n+1\\
f_{n+1,m} &= &f_{n,m} + f_{n,m+1} + \\
&& \sum_{p=0}^{n} \sum_{c=0}^{m} \sum_{k=0}^{m - c} {m \choose c} {m - c\choose k} f_{p,k+c} f_{n-p,m-k}.
\end{eqnarray*}
This formula is the most complex.  
Appendix~\ref{sec:formula} shows how it is constructed.



\paragraph{$0$-contexts}~
\label{sec:0-contexts}

$c_{n,0}$ where
  \begin{eqnarray*}
    c_{0,i} &=& [i = 1]\\
    c_{n+1,i} & = & \sum_{j=i}^{n+1} {j \choose i}\,c_{n,j} + \sum_{j=0}^i\sum_{k=0}^nc_{k,j} c_{n-k,i-j}.
  \end{eqnarray*}

\paragraph{Four relations.}
\label{sec:four-relations}

\begin{sloppypar}
  Let us use the notation $R^{(m)}_i$ (see \ifJFP\else{Flajolet and Sedgewick's book \cite{flajolet08:_analy_combin}}\fi) for the
  number of surjections from $[1..i]$ to $[1..m]$.  Recall that
\end{sloppypar}
\[R^{(m)}_i \quad = \quad \sum_{j=0}^i {i \choose j} (-1)^j (i-j)^m.\]

The numbers $T_{n,m}$, $f_{n,m}$ and $c_{n,i}$ are related as follows (see
Appendix~\ref{sec:relat-betw}):
\begin{displaymath}
  \begin{array}{lclcl}
    T_{n,m} &=& \displaystyle \sum_{i=0}^{m} {m \choose i} f_{n,i} & = &
    \displaystyle \sum_{i=1}^{n+1} c_{n,i}    m^i\\\\
    f_{n,m} &=& \displaystyle \sum_{i=0}^m (-1)^{m+i} {m \choose i} T_{n,i} & = &
    \displaystyle \sum_{i=1}^{n+1}     c_{n,i} R^{(m)}_i.
  \end{array}
\end{displaymath}

\section{More generating functions}
\label{sec:generating}

In this section, we provide the asymptotic approximation of the growth of the
$k^{th}$ coefficients of the polynomials $P_n$, where the first coefficient is the
coefficient of the monomial of highest degree.

For every positive integer $i$, let us denote by $a_i$ the generating function for the sequence $\left( \p{i}{n} \right)_{n \geq 0}$, \ie
\[ a_i(z) = \sum_{n=0}^{\infty} \p{i}{n} z^n = \sum_{n=0}^{\infty} c_{n,n+2-i} z^n .\] The
$\p{i}{n}$'s count the number of contexts of size $n$ having $n+2-i$ holes. For the sake of clarity, instead of writing $a_i(z)$ sometimes we simply write
$a_i$.

In order to compute the functions $a_i$, we apply the following basic fact about generating functions.

\begin{fact}\label{basicGen}
Let $f$ and $g$ be generating functions for sequences $\left( f_n \right)_{n \geq 0}$ and $\left( g_n \right)_{n \geq 0}$, respectively. Then
\begin{enumerate}[(i)]
\item the generating function for the sequence $\left( {n \choose k}f_n \right)_{n \geq 0}$, where $k$ is a fixed positive integer, is given by $\frac{z^k f^{(k)}}{k!}$,
\item the generating function for the sequence $\left( \sum_{i=0}^n f_i g_{n-i} \right)_{n \geq 0}$ is given by $f \cdot g$,
\item the generating function for the sequence $\left( {n-j \choose i} f_n \right)_{n \geq 0}$, where $i \geq 0$ and $j > 0$, is given by $\sum_{k=0}^{i} (-1)^k {k+j-1 \choose j-1}z^{i-k} \frac{f^{(i-k)}}{(i-k)!}$.
\end{enumerate}
\end{fact}

\begin{proof}
Items (i) and (ii) can be found, e.g., in Chapter 7 of \cite{GraKnuPat}.

The third part follows from (i) and the following equality:
\[{n-j \choose i} = \sum_{k=0}^{i} (-1)^k {n \choose i-k} {k+j-1 \choose j-1},\]
which holds for every $n,i \geq 0$ and $j > 0$. This equality can be easily derived from
two equalities known as ``upper negation'' and ``Vandermonde convolution'', which can be
found in Table 174 of \cite{GraKnuPat}.
\end{proof}

Now we are ready to provide a recurrence for functions $a_i$.

\begin{theo}\label{functionsRel}
The following equations are valid:
\begin{eqnarray*}
 a_1 &=& za_1^2 + 1, \quad a_1(0)=1\\
 a_2 &=& za_1 + 2za_1a_2\\
 a_i &=& z^{i-1}\frac{a_1^{(i-2)}}{(i-2)!} + z^{i-2}\frac{a_1^{(i-3)}}{(i-3)!} + z^{i-2}\frac{a_2^{(i-3)}}{(i-3)!} \\
      &&+ z \cdot \sum_{j=1}^{i-3} \sum_{k=0}^{i-3-j} (-1)^k {k+j-1 \choose j-1} z^{i-3-j-k} \frac{a_{j + 2}^{(i-3-j-k)}}{(i-3-j-k)!} \\
      &&+ z \cdot \sum_{j=1}^{i} a_j a_{i-j+1}, \qquad \textit{~for~}  i > 2.
 \end{eqnarray*}
\end{theo}


\begin{proof}
All these equations follow from Lemma~\ref{p-values} and Fact~\ref{basicGen}.
\end{proof}

Notice that the $a_i$'s can be computed by induction. Indeed, $a_i$ occurs twice in
the last sum on the right hand side of the last equation and we have:
\begin{eqnarray*}
  a_i (1-2a_1z) &=& z^{i-1}\frac{a_1^{(i-2)}}{(i-2)!} + z^{i-2}\frac{a_1^{(i-3)}}{(i-3)!} + z^{i-2}\frac{a_2^{(i-3)}}{(i-3)!} \\
      &&+ z \cdot \sum_{j=1}^{i-3} \sum_{k=0}^{i-3-j} (-1)^k {k+j-1 \choose j-1} z^{i-3-j-k} \frac{a_{j + 2}^{(i-3-j-k)}}{(i-3-j-k)!} \\
      &&+ z \cdot \sum_{j=2}^{i-1} a_j a_{i-j+1}.
\end{eqnarray*}
Since $1-2a_1 z= \sqrt{1-4z}$, we get:
\begin{displaymath}
a_i = \left(
  \begin{array}{l}
    z^{i-1}\frac{a_1^{(i-2)}}{(i-2)!} + z^{i-2}\frac{a_1^{(i-3)}}{(i-3)!} + z^{i-2}\frac{a_2^{(i-3)}}{(i-3)!} \\
      + z \cdot \sum_{j=1}^{i-3} \sum_{k=0}^{i-3-j} (-1)^k {k+j-1 \choose j-1} z^{i-3-j-k} \frac{a_{j + 2}^{(i-3-j-k)}}{(i-3-j-k)!} \\
      + z \cdot \sum_{j=2}^{i-1} a_j a_{i-j+1}
    \end{array}
  \right) / \sqrt{1-4z}. \qquad (\ddagger)
    \end{displaymath}


\begin{corollary}\label{functions}
Exact formulas for the functions $a_1$--$a_7$ are given in Figure \ref{fig:gen}.
\begin{figure*}
\newcommand{\RHOF}{\sqrt{1-4z}}
  \centering
    \begin{eqnarray*}
      a_1(z) &=& \left( \frac{1}{2}-\frac{(1-4z)^{1/2}}{2}\right) z^{-1}\\
      a_2(z) &=& -\frac{1}{2} + \frac{1}{2\;(1-4z)^{1/2}}\\
      a_3(z) &=& \left( \frac{1}{1-4z} + \frac{z}{(1-4z)^{3/2}} \right) z\\
      a_4(z) &=& \left( \frac{3}{(1-4z)^2} + \frac{z}{(1-4z)^{5/2}} \right) z^2\\
      a_5(z) &=& \left( \frac{4z+9}{(1-4z)^3} + \frac{z^2 - 19z + 5}{(1-4z)^{7/2}} \right) z^3\\   	
      a_6(z) &=& \left( \frac{24z+31}{(1-4z)^4} + \frac{3z^2 - 203z + 51}{(1-4z)^{9/2}} \right) z^4\\
      a_7(z) &=& \left( \frac{16z^2 - 128z + 181}{(1-4z)^5} + \frac{2z^3 - 194z^2 - 1541z + 398}{(1-4z)^{11/2}} \right) z^5
    \end{eqnarray*}
    \caption{The generating functions for the coefficients of the polynomials $P_n(m)$}
\label{fig:gen}
\end{figure*}
\end{corollary}

\begin{proof}
Let us first compute the function $a_1$ which, according to Theorem \ref{functionsRel}, is given by
\[a_1 = za_1^2 + 1, \quad a_1(0)=1.\] By solving this equation, we obtain $a_1(z) =
\frac{1-\sqrt{1-4z}}{2z}$, which is exactly the generating function for Catalan
numbers---see, e.g., Chapter I.1 of \cite{flajolet08:_analy_combin}.

Now, let us notice that on the basis of Theorem \ref{functionsRel} all the other functions can be immediately obtained by tedious, however elementary, computations. In order to get exact values we applied \textsf{Sage} software \cite{sage}.
\end{proof}

Let $[z^n]f(z)$ denote the $n^{th}$ coefficient of $z^n$ in the formal power series
$f(z) = \sum_{n=0}^{\infty} f_n z^n$. As usual, we use the symbol $\sim$ to denote
the asymptotic equivalence of two sequences, i.e., we write $f_n \sim g_n$ iff the
limit of the sequence $(f_n/g_n)_{n\geq 0}$ is $1$. Similarly, by $f(z) \simtozero g(z)$ we mean that the limit of $f(z)/g(z)$ is $1$ when $z \to z_0$.
We say that a function $f(z)$ is of order $g(z)$ for $z \to z_0$ iff there exists a
positive constant $A$ such that $f(z) \simtozero A\cdot g(z)$.

The theorem below (Theorem VI.1 of \cite{flajolet08:_analy_combin}) serves as a powerful tool that allows us to estimate coefficients of certain functions that frequently appear in combinatorial considerations.

\begin{fact}\label{fact:asym_exp}
Let $\alpha$ be an arbitrary complex number in $\mathbb{C}\setminus \mathbb{Z}_{\leq 0}$. The coefficient of $z^n$ in
\[ f(z) = (1-z)^{\alpha} \]
admits the following asymptotic expansion:
\begin{eqnarray*}
[z^n]f(z) &\sim& \frac{n^{\alpha - 1}}{\Gamma (\alpha)} \left( 1 + \frac{\alpha (\alpha -
    1)}{2n}  + \frac{`a(`a-1)(`a-2)(3`a-1)}{24n^2}\right. \\
&& \qquad\qquad \left. + \frac{`a^2(`a-1)^2(`a -2)(`a-3)}{48n^3}+ O \left( \frac{1}{n^{4}} \right) \right) ,
\end{eqnarray*}

where $\Gamma$ is the Euler Gamma function defined for $\Re(\alpha) > 0$ as
\[ \Gamma ( \alpha ) := \int_{0}^{\infty} e^{-t} t^{\alpha -1 } dt .\]
\end{fact}





Now we are ready to prove the following approximation.
\begin{prop}\label{prop:equiv}
The exact order of functions $a_i$ for $z \to 1/4$ is given by
  \begin{displaymath}
  a_i(z) \quad \subrel{z \to \frac{1}{4}}{\sim} \quad \frac{C_{i-2}}{2^{3i-5}(1-4z)^{(2i-3)/2}},
\end{displaymath}
where $C_i$ is the $i^{th}$ Catalan number.
\end{prop}
\begin{proof}
We prove the result by induction using Theorem~\ref{functionsRel}. For the sake of simplicity, we write $\sim$ and ``is of order'' to denote $\subrel{z \to \frac{1}{4}}{\sim}$ and ``is of order for $z \to 1/4$''.

The result is true for $i=1$. For $i>1$ and $j\le i$,  assume that $a_j(z)$ is of order $\frac{1}{(1-4z)^{(2j-3)/2}}$ and look at
  equation $(\ddagger)$ to prove that  $a_{i+1}(z)$ is of order $\frac{1}{(1-4z)^{(2i-1)/2}}$.

  Notice that the $i^{th}$ derivative of $a_1$ is of order
  $\frac{1}{(1-4z)^{(2i-1)/2}}$, hence its $(i-2)^{th}$ derivative is of order
  $\frac{1}{(1-4z)^{(2i-5)/2}}$ and its $(i-3)^{th}$ derivative is of order
  $\frac{1}{(1-4z)^{(2i-7)/2}}$.  Similarly, the $i^{th}$ derivative of $a_2$ is of
  order $\frac{1}{(1-4z)^{(2i+1)/2}}$, hence its $(i-3)^{th}$ derivative is of order
  $\frac{1}{(1-4z)^{(2i-5)/2}}$.

  By induction for $j+2\le i-3$,  $a_{j+2}$ is of order $\frac{1}{(1-4z)^{(2j+1)/2}}$.
  Among its successive derivatives we derive at most $i-3-j$ times, hence the
  items in the sum are of order at most $\frac{1}{(1-4z)^{(2i-5)/2}}$.

  Now, every product $a_j a_{i-j+1}$ is of order $\frac{1}{(1-4z)^{i-2}}$, therefore the first four terms in $(\ddagger)$ do not
  contribute to the asymptotic value of $a_{i+1}(z)$.  Hence the contribution to the asymptotic value is given only by products $a_j a_{i-j+1}$'s. Multiplying their order by $\frac{1}{\sqrt{1-4z}}$ we obtain that the last sum is of order $\frac{1}{(1-4z)^{(2i-1)/2}}$.

    Let us denote by $K_i$ the multiplicative coefficient $C_{i-2}/2^{3i-5}$ of $\frac{1}{(1-4z)^{(2j-3)/2}}$.
    One notices that $K_2= \frac{1}{2} = \frac{C_0}{2^{3\times 2-5}}$.  The sum
    $z~\sum_{j=2}^{i-1} a_j a_{i-j+1}$ shows the inductive part.  Indeed, when
    $z=\frac{1}{4}$:
\begin{eqnarray*}
  z~\sum_{j=2}^{i-1} K_j K_{i-j+1} &=&
  \frac{1}{4} ~\sum_{j=2}^{i-1} \frac{C_{j-2}}{2^{3j-5}}\ \frac{C_{i-j+1-2}}{2^{3(i-j+1)-5}} \\
&=& \frac{1}{2^{3i-5}} \ \sum_{j=0}^{i-3} C_{j} C_{i-j-3}\\
&=& \frac{C_{i-2} }{2^{3i-5}} \ = \ K_i.
\end{eqnarray*}
\end{proof}

Finally, we are able to provide asymptotic values of coefficients of functions $a_i$.

\begin{theo}
The coefficient of $z^n$ in the function $a_k(z)$ admits the following asymptotic expansion:
    \begin{eqnarray*}
      [z^n] a_k(z) &= &\frac{1}{2^{k-1} (k-1)! \sqrt{`p}} \ 4^nn^{(2k-5)/2} \ \cdot \ `J(n,k)
  \end{eqnarray*}
where
\begin{eqnarray*}
  `J(n,k)&=& \quad 1 + \frac{(2k-3)(2k-5)}{8n} +
  \frac{(2k-3)(2k-5)(2k-7)(3k-11)}{384n^2} + \\
&& \quad \frac{(2k-3)^2(2k-5)^2(2k-7)(2k-9)}{3672n^3}
  + O\Big(\frac{1}{n^4}\Big).
\end{eqnarray*}
\end{theo}
  \begin{proof}
First recall that
\[`G((2k-3)/2) = `G\Big((k-2)+\frac{1}{2}\Big) = \frac{(2(k-2))!\sqrt{`p}}{2^{2(k-2)} (k-2)!}.\]
Now using Fact~\ref{fact:asym_exp}, we can compute the principal  part:
\begin{eqnarray*}
  [z^n]a_k(z) &=& \frac{C_{k-2}}{2^{3k-5}}\ 4^n \ [z^n](1-z)^{(2k-3)/2}\\
  &\sim&\frac{C_{k-2}}{2^{3k-5}}\ 4^n \ \frac{n^{(2k-5)/2}}{`G((2k-3)/2)}\\
  &=& \frac{C_{k-2}}{2^{3k-5}} \frac{(k-2)! 2^{2(k-2)}}{(2(k-2))! \sqrt{`p}}\
  4^nn^{(2k-5)/2}\\
&=& \frac{C_{k-2} (k-2)!}{2^{k-1} (2(k-2))! \sqrt{`p}} \ 4^nn^{(2k-5)/2}\\
&=& \frac{1}{2^{k-1} (k-1)! \sqrt{`p}} \ 4^nn^{(2k-5)/2}.
\end{eqnarray*}
For $`J(n,k)$ we use Fact~\ref{fact:asym_exp} with $`a = \frac{2k-3}{2}$.
  \end{proof}


By looking at Figure \ref{fig:gen}, we can easily notice a recurring pattern concerning the structure of functions $a_i$. Therefore, we state the following proposition.

\begin{prop}\label{prop:Q_R}
For every $i > 2$ we have
\[a_i(z) = z^{i-2} \left( \frac{Q_i(z)}{(1-4z)^{i-2}} + \frac{S_i(z)}{(1-4z)^{i-\frac{3}{2}}} \right), \]
where $Q_i$ and $S_i$ are polynomials over ${\mathbb Z}$ in $z$ and $\deg Q_i = \left\lfloor \frac{i-3}{2} \right\rfloor$ and $\deg S_i = \left\lfloor \frac{i-1}{2} \right\rfloor$.
\end{prop}
\begin{proof}
  By induction using  formula $(\ddagger)$, in the same vein as the proof of
  Proposition~\ref{prop:equiv}.  In particular, the two first members of $(\ddagger)$ are
  derivatives of the generating function of Catalan numbers studied in \cite{lang02:_polyn_cataly}.
\end{proof}

  As we have already mentioned, the number of closed terms of size $n$ is given by
  $P_n(0)$, which corresponds to the $n^{th}$ term of the Taylor expansion of the function
  $a_{n+2}$. Hence, the sequence of the numbers of closed $`l$-terms is equal to the
  sequence $\left( [z^n]a_{n+2}(z) \right)_{n \geq 0}$. From Proposition~\ref{prop:Q_R}, the
  number of closed terms of size $n$ is equal to $Q_{n+2}(0) + S_{n+2}(0)$.  Currently, we
  have no recursive formula for the $Q_n$'s and the $S_n$'s.  However, by Proposition~\ref{prop:equiv}, we know that
\[S_{n+2}\left(\frac{1}{4}\right) = \frac{C_{n}}{2^{n+1}}.\]

\section{Counting normal forms}
\label{sec:count-norm-forms}

Beside counting terms, it is also interesting to count normal forms.  To this end, we
describe the set of normal forms as follows
\begin{eqnarray*}
  \G_{n+1,m} &=&  [\Var{1} .. \Var{m}] \ \uplus\ \biguplus_{i=0}^n \G_{i,m} @ \F_{n-i,m}\\
  \F_{n+1,m} &=& `l\,\F_{n,m+1} \uplus \G_{n,m}
\end{eqnarray*}
Recall that a normal form consists of a (possibly empty) sequence of abstractions
followed by the application of a de Bruijn index to normal forms.  $\F_{n,m}$
represents the normal forms of size $n$ with at most $m$ free indices and $\G_{n,m}$
represents the neutral terms, i.e., terms starting with an index, of size $n$ with at
most $m$ free indices.  From this we derive the formulas for counting:
\begin{eqnarray*}
  G_{0,m} &=& m\\
  G_{n+1,m} &=& \sum_{k=0}^{n}G_{n-k,m} F_{k,m},\\\\
  F_{0,m} &=& m\\
  F_{n+1,m}&=& F_{n,m+1} + G_{n+1,m}.
\end{eqnarray*}
The values of $F_{n,0}$ up to $n=10$ are:
\[0, 1, 3, 11, 53, 323, 2359, 19877, 188591, 1981963, 22795849.\] This sequence, added by us to the \emph{On-line Encyclopedia of Integer Sequences}, has its entry number \textbf{A224345}.  A
\textsf{Haskell} program for computing the values of $F_{n,m}$ and
$G_{n,m}$ is given in Figure~\ref{fig:nf-program}.
\begin{figure}[!htb]
  \centering
\hrule

\medskip

\begin{verbatim}
ftab :: [[Integer]]
ftab = [0..] : [[f' (n-1) (m+1) + g' n m | m<-[0..]] | n<-[1..]]

gtab :: [[Integer]]
gtab = [0..] : [[s n m |  m <- [0..]] | n <- [1..]]
  where s n m  = let fi =  [f' i m | i <- [0..(n-1)]]
                     gi =  [g' i m | i <- [n-1,n-2..0]]
                 in sum $ zipWith (*) fi gi
 
f' :: Int -> Int -> Integer
f' n m = ftab !! n !! m
g' n m = gtab !! n !! m
\end{verbatim}
 \hrule
  \caption{\textsf{Haskell} program for counting normal forms}
\label{fig:nf-program}
\end{figure}

The efficiency of this program can be improved (Figure~\ref{fig:nf-imp-program}).
\begin{figure}[!htb]
  \centering
\hrule

\medskip

\begin{verbatim}
fgtab :: [[[(Integer,Integer)]]]
fgtab = iterate nextn . map return $ zip [0..] [0..]
  where
  nextn ls = zipWith rake (tail ls) ls
  rake ((f1,_):_) ms = let cv = conv ms in (f1 + cv, cv) : ms
  conv ms = sum $ zipWith (\(a,_) -> \(_,b) -> a*b) ms (reverse ms)

f :: Int -> Int -> Integer
f n m = fst $ head $ fgtab !! n !! m
g n m = snd $ head $ fgtab !! n !! m
\end{verbatim}
\hrule
  \caption{\textsf{Haskell} improved program for counting normal forms}
  \label{fig:nf-imp-program}
\end{figure}
 Like for terms we derive polynomials:
\begin{eqnarray*}
  \PNF_0(m) &=& m \\
  \PNF_{n+1}(m) &=& \PNF_n(m+1) + \QNF_{n+1}(m),\\\\
  \QNF_0(m) &=& m \\
  \QNF_{n+1}(m) &=& \sum_{k=0}^{n} \PNF_k(m) \QNF_{n-k}(m).
\end{eqnarray*}
\begin{lemma}
For every $n$, the degree of the polynomials $\PNF_n$ and $\QNF$ is equal to $n+1$.
\end{lemma}
\begin{proof}
  Like the proof of Lemma~\ref{lem:degree},  by induction on $n$ from the definition of
  $\PNF_n$ and $\QNF$.
\end{proof}

\subsection{Coefficients of the polynomials $\PNF_n$ and $\QNF_n$}

Let us count $i$-nf-contexts.  They are closed normal forms with $i$ holes.  The
$i$-nf-contexts of size $n$ are counted by $d_{n,i}$. They are abstractions of
$i$-contexts of the form $[\ ]\,N_1\ldots N_p$, which we call $i$-nf-pre-contexts,
where each $N_j$ is a $i_j$-nf-context (with $i_1+\ldots+i_j+\ldots+i_p = i-1$) and
which are counted by $g_{n,i}$.  There is one $1$-nf-context and one
$1$-nf-pre-context of size $0$, whereas there are $0$ $i$-nf-contexts and $0$
$i$-nf-pre-contexts for $i\neq 1$ of size $0$. Thus we get
\begin{eqnarray*}
  d_{0,i} &=& [i = 1],\\
  g_{0,i} &=& [i = 1].
\end{eqnarray*}

By reasoning similarly as in Section~\ref{sec:an-comb-interpr} and by using the
description of normal forms given above, we get:
\begin{eqnarray*}
  d_{n+1,i} &=&  \sum_{j=i}^{n+1} {j \choose i}\,d_{n,j} + g_{n+1,i}, \\
  g_{n+1,i} &=& \sum_{j=0}^i\sum_{k=0}^n g_{k,j} d_{n-k,i-j}.
\end{eqnarray*}

Therefore 
\begin{eqnarray*}
\PNF_n(m)  &=& \sum_{i=0}^n d_{n,i} m^i,\\
\QNF_n(m)  &=&  \sum_{i=0}^n g_{n,i} m^i.
\end{eqnarray*}

\subsection{Generating functions}
\label{sec:D_and_G}
Consider the two generating functions:
\begin{eqnarray*}
  D(z,u) &=& \sum_{n,i\ge 0}d_{n,i}z^n u^i,\\
  G(z,u) &=& \sum_{n,i\ge 0}g_{n,i}z^n u^i.
\end{eqnarray*}
Then we have
\begin{eqnarray*}
  D(z,u) &=& \sum_{n=0}^{\infty}\PNF_n(u) z^n,\\
  G(z,u) &=& \sum_{n=0}^{\infty}\QNF_n(u) z^n.
\end{eqnarray*}
Therefore
\begin{eqnarray*}
  D(z,u) &=& u + z \sum_{n=0}^{\infty} \PNF_n(u+1) z^n +  \sum_{n=1}^{\infty}
  \QNF_n(u) z^n\\
  &=& z D(z,u+1) + G(z,u)
\end{eqnarray*}
and
\begin{eqnarray*}
  G(z,u) &=& u + z \sum_{n=0}^{\infty} \QNF_n(u) \PNF_n(u) z^n\\
  &=& u + z \sum_{n=0}^{\infty} \sum_{k=0}^n g_{k,j} z^{k} u^i d_{n-k,i-j} z^{n-k}n u^{i-j}\\
  &=& u + z D(z,u)\, G(z,u).
\end{eqnarray*}
Consequently the two functions $D$ and $G$ satisfy 
\begin{eqnarray*}
    D(z,u) &=& z D(z,u+1) + G(z,u),\\
    G(z,u) &=& u + zD(z,u)\, G(z,u).
\end{eqnarray*}
$D(z,0)$ is the generating function for the numbers of closed normal forms of size
$n$.
By solving the above system of equations, we get:
\begin{displaymath}
  z D(z,u)² - (1+z^2 D(z,u+1)) D(z,u)  + u + z D(z,u+1) = 0,
\end{displaymath}
which yields
\begin{displaymath}
  D(z,u) = \frac{1 + z^2 D(z,u+1) - \sqrt{(1+z^2 D(z,u+1))^2 - 4z (u +z D(z,u+1))}}{2z}.
\end{displaymath}

\section{Lambda term generation}
\label{sec:lambda-term-gener}

From the simple equation defining the number $T_{n,m}$ of terms, we define the
function generating them.  More precisely, we define a function \textsf{unrankT~n~m~k}
which returns the $k^{th}$ term of size $n$ with at most $m$ distinct free variables
(see the \textsf{Haskell} program in Figure~\ref{fig:prog-gen}).  The variable $k$~is
an \textsf{Integer} (\ie an arbitrary-precision integer) which belongs to the interval $[1
.. T_{n,m}]$.  The unranking program mimics counting terms. If $n$ is $0$, then the
program returns the de Bruijn index $\Var{k}$.  Otherwise, if $k$ is less than
$T_{n-1,m+1}$, the rank $k$ lies in the part of the interval $[1 .. T_{n,m}]$ with
terms that are abstractions. Therefore, for $k\le T_{n-1,m+1}$ \textsf{unrankT~n~m~k}
returns $`l\,$\textsf{(unrankT~(n-1)~(m+1)~k)}. If the rank $k$ is larger than
$T_{n-1,m+1}$, it lies in the part of the interval $[1 .. T_{n,m}]$ with
applications.  Therefore we call a function \textsf{appTerm} which tries to identify which
sub-interval contains a pair of terms with indices $k'$ and $k''$ such that $k'+ k''$
is at the right place. The product of these values correspond to one of the products
$T_{j,m}T_{n-j,m}$ in the sum.  When the number~$j$ is found, two recursive calls of
\textsf{unrankT}, with appropriate $k'$ and $k''$, build the subterms of the application.  One
may notice $(h-1)$ and $+1$ which take into account the fact that $k$ lies in an
interval $[1..T_{\_,\_}]$ while \textsf{divMod} works in an interval
$[0..(T_{\_,\_}-1)]$.

The function \textsf{unrankT} relies on the function $t$ presented in
Section~\ref{sec:actual} and called here $O(n)$ times.  Assuming that $t$ has been
called once already and therefore runs in $O(n+m)$, \textsf{unrankT} performs $O(n)$
recursive calls and its complexity depends on one side linearly on the operations
\textsf{divMod}, $(-)$ and $(*)$ performed on arbitrary-precision integers and on the other side is in
$O(n^2)$ due to the accesses generated by $t$.

For a given $n$, this program can be used to enumerate all the closed $`l$-terms of size $n$ and, more
generally, all the $`l$-terms of size $n$ with at most $m$ distinct free variables.
This is appropriate only for small values of $n$, since the number of $`l$-terms gets superexponentially 
large with~$n$.  But overall, in order to generate a random term of size~$n$ with at most $m$
distinct free variables, it suffices to feed $T$ with a random value $k$ in the
interval $[1..T_{n,m}]$.  Similarly, on the basis of the recursive formula for the
number of normal forms, one defines a program for their generation (Figure~\ref{fig:nf}).

\begin{figure*}
\centering
\hrule

\medskip

\begin{verbatim}
data Term = Index Integer
          | Abs Term
          | App Term Term

unrankT :: Int -> Int -> Integer -> Term
unrankT 0 m k = Index k
unrankT n m k
    | k <= (t (n-1) (m+1)) = Abs (unrankT (n-1) (m+1) k)
    | (t (n-1) (m+1)) < k = appTerm (n-1) 0 (k - t (n-1) (m+1))
    where appTerm n j h
            | h <=  tjmtnjm  = let (dv,md) = ((h-1) `divMod` tnjm) 
                               in App (unrankT j m (dv+1)) 
                                      (unrankT (n-j) m (md+1))
            | otherwise = appTerm n (j + 1) (h -tjmtnjm) 
            where tnjm = t (n-j) m
                  tjmtnjm = (t j m) * tnjm 
\end{verbatim}
\hrule
\caption{\textsf{Haskell} program for term unranking}
\label{fig:prog-gen}
\end{figure*}

\begin{figure*}[!htb]
  \centering
\hrule
\medskip
\begin{verbatim}
unrankNF :: Int -> Int -> Integer -> Term
unrankNF 0 m k = Index k
unrankNF n m k
  | k <= f (n-1) (m+1) = Abs (unrankNF (n-1) (m+1) k)
  | f (n-1) (m+1) < k = unrankNG n m (k - f (n-1) (m+1))

unrankNG :: Int -> Int -> Integer -> Term
unrankNG 0 m k = Index k
unrankNG n m k = appNF (n-1) 0 m k
                                   
appNF :: Int -> Int -> Int -> Integer -> Term
appNF n j m h
    |  h <=  gjmfnjm  = let (dv,md) = (h-1) `divMod` fnjm
                        in App (unrankNG j m (dv+1)) 
                               (unrankNF (n-j) m (md +1))
    | otherwise = appNF n (j + 1) m (h -gjmfnjm)
    where fnjm = f (n-j) m                            
          gjmfnjm = g j m * fnjm
\end{verbatim}
\hrule
  \caption{\textsf{Haskell} program for normal form unranking}
  \label{fig:nf}
\end{figure*}

\section{Simply typable terms}
\label{sec:simply-typed-terms}

Once we have a random generator for untyped terms, it is easy to build a random generator
for simply typable terms. It suffices to sieve all terms by a predicate, which we call
\emph{isTypable}. This predicate is a classical principal type
algorithm
\cite{newman43:_strat,DBLP:conf/popl/DamasM82,DBLP:journals/logcom/Hindley08}.  In
Appendix~\ref{sec:typ_prog}, we give a \textsf{Haskell} program.
For instance,
applying the random generator with parameter~$10$ (for the size of the term), we got:
\begin{displaymath}
  `l (`l (((\Var{1}\  `l (\Var{1}))\  `l ((\Var{3}\  `l (((\Var{1}\  \Var{2})\  \Var{3}))))))) .
\end{displaymath}
This is a ``typical'' simply typable random closed $`l$-term of size $10$ written with de
Bruijn indices.  Its type is

\begin{center}
  \begin{math}
    ((`a "->" (((`b "->" `b) "->" (`a "->" `g) "->" `d) "->" `z)) "->" `z) "->" `g "->"
   \end{math}

    \begin{math}
    ((`b "->" `b) "->" (`a "->" `g) "->" `d) "->" `d .
  \end{math}
\end{center}

We were able to generate typable terms of size $50$.
For such terms, the generating process is slow, since it requires $50\,000$ generations of
terms, with (unsuccessful) tests of their typability before getting a typable one.  But for
size $40$, 
the number of attempts falls to $3$ for $10\,000$.

This kind of random generator is useful for testing functional programs. Micha{\l}
Pa{\l}ka \cite{palka12:_testin_compil,Palka:2011:TOC:1982595.1982615} proposed a tool
to debug Haskell compilers based on a $`l$-term generator.  His generator is
designed on the development of a typing tree, with choices made when a new rule is
created.  Such a method needs to cut branches in developing the tree to avoid loops.
This way his generator is not random, which may be a drawback in some cases.  As a
matter of fact, a method for generating simply typed terms based on developing a
typing tree does not produce terms on a uniform random distribution since it requires
to cut the tree at arbitrary locations to avoid loops, ``arbitrary'' in the sense of
randomness preservation.  In other words, there is no simple recursive definition of
simply typed terms, as well as of simply typable terms,  that would allow an easy uniform random generation.  This is also
what makes the combinatorial study of typed terms difficult.  A term is typable
because it satisfies some constraints, not because it is generated in a specific way.

\section{Experimental data}
\label{sec:exper-data}

Given a random term generator, we are able to write programs to make statistics on some
features of terms.  While there are many possible experiments of this type, here we
present only two that we find interesting and suggestive of other possibilities. 



\subsection{Average variable depth in closed terms and closed normal forms}
\label{sec:avearge-depth-var}

\begin{figure}[htb!]
  \centering
  \includegraphics[width=0.9\textwidth]{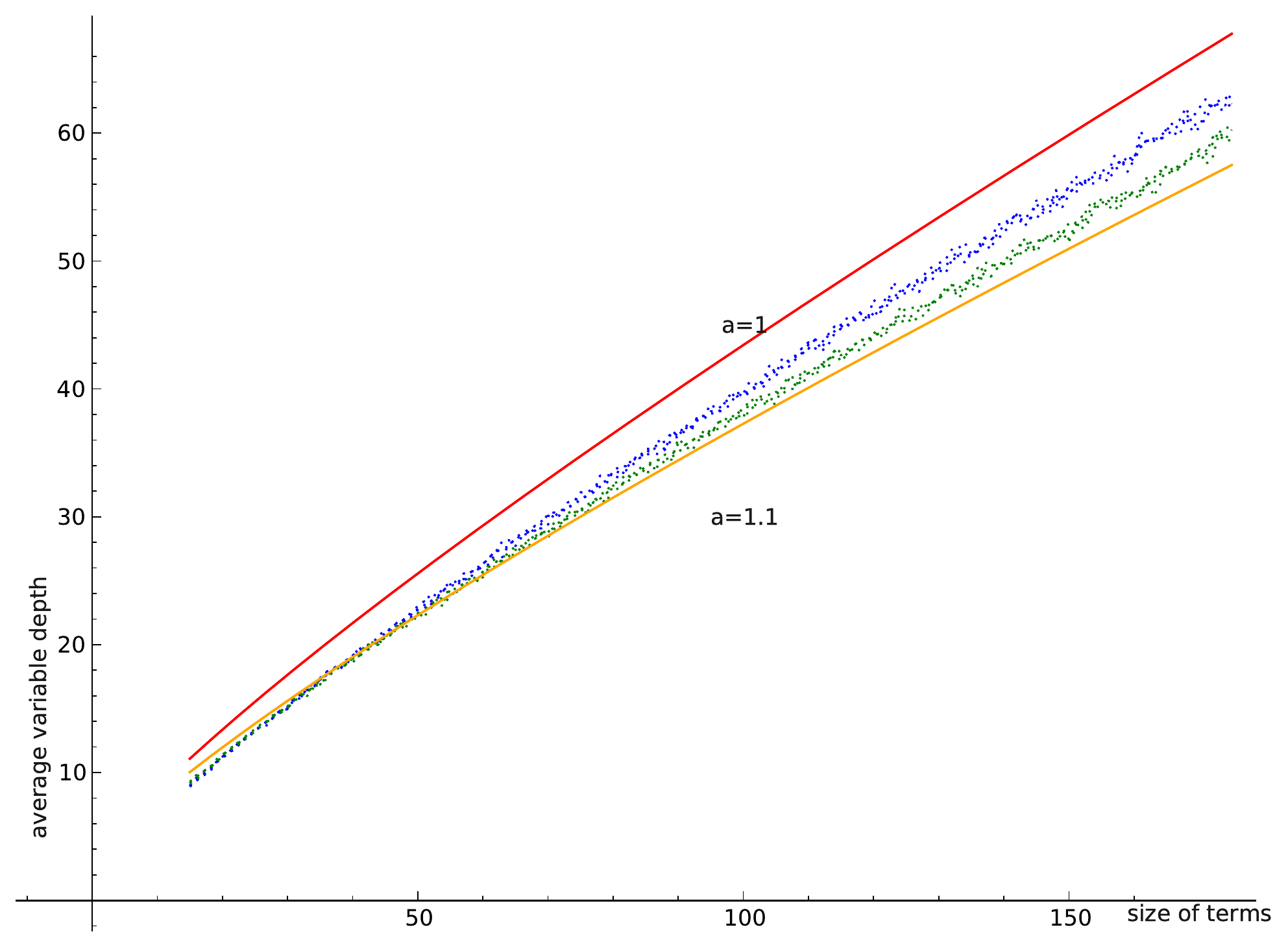}
  \caption{\textsf{From top to bottom:} Curve $\frac{2n}{\ln(n)}$, average variable
    depth for closed terms, average variable depth for closed normal forms and curve
    $\frac{2n}{\ln(n)^{1.1}}$.}
  \label{fig:av_depth}
\end{figure}

\begin{figure}[htb!]
  \centering
  \includegraphics[width=0.9\textwidth]{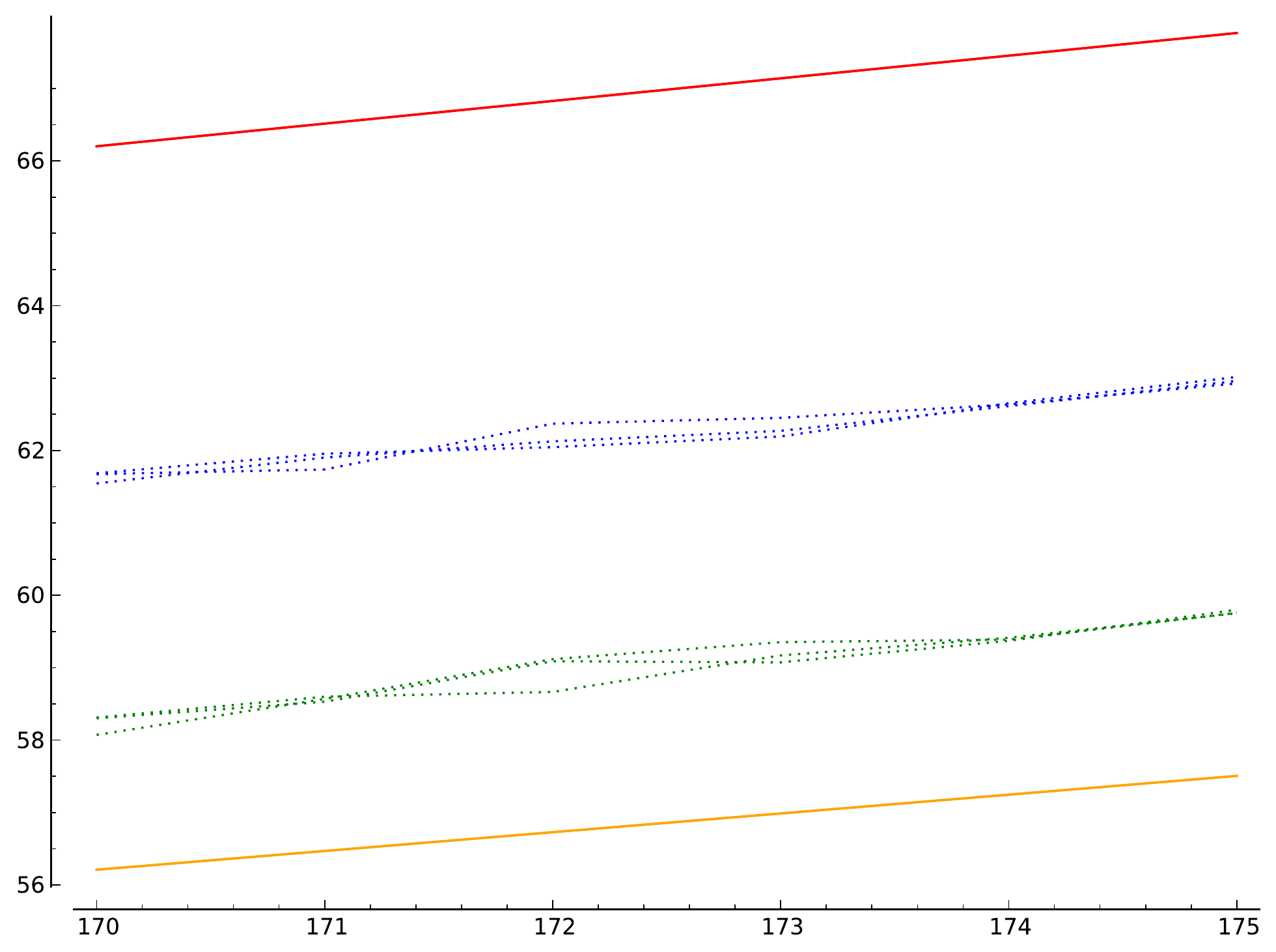}
  \caption{Magnification of Figure~\ref{fig:av_depth} between $n=170$ and $n=175$.}
  \label{fig:av_depth_170_175}
\end{figure}

Let us define the \emph{variable depth} as the number of symbols (abstractions and
applications) between a variable and the top of the term. For instance, given the
term $\lambda x.(\lambda y z.x) (\lambda u.u)$, the first occurrence of variable $x$
has depth~$1$ and the second occurrence of variable $x$ has depth $3$, while the
depth of $u$ is~$2$.  This gives the average depth~$2$ for this term.  Looking at the de
Bruijn indices of the brother term $`l \Var{1} \,(`l `l \Var{3}\, `l \Var{1})$, we
say that the first index~$\Var{1}$ has depth $1$, the second index $\Var{3}$ has
depth $3$ and the third index $\Var{1}$ has depth~$2$, with the same average $2$ as
previously.  In Figure~\ref{fig:av_depth}, we draw the average variable depth for
$300$ random closed terms of size~$15$ up to size $175$ (top scatter plot) and the average
variable depth for $300$ random normal forms of size~$15$ up to size $175$ (bottom
scatter plot) squeezed between the curves $\frac{2n}{\ln(n)^a}$ for $a=1$ and $a=1.1$
(plain lines).  In Figure~\ref{fig:av_depth_170_175} we see the same four curves
enlarged in the interval $[170..175]$.  This shows clearly that the
average variable depth of closed terms and closed normal forms are different.  On this basis, we
conjecture that the average depth of variables in closed terms is asymptotically bounded from above by
$\frac{2n}{\ln(n)}$ and that the average variable depth is slightly smaller for
normal forms than for closed terms.

\subsection{Average number of head $`l$'s per closed term}
\label{sec:average-number-l}

We say that $\lambda x$ is a \emph{head lambda} in a term $t$ if the latter is of the
form $\lambda x_1 \ldots \lambda x_n \lambda x.s$ for some positive integer $n$ and a
certain term $s$.  In order to know the structure of an average term, we are
interested in the average number of head $`l$'s occurring in closed terms.  In
Figure~\ref{fig:L-in-terms}, we compare values of some functions $\sqrt{\frac{n}{\ln(n)^a}}$ 
with the number of head $`l$'s in 1000 random closed terms and
the average number of head $`l$'s in 1000 normal forms, both in the case when size
goes from~$15$ to $150$.  We see that, in the case of closed terms, these numbers are
in accordance with Theorem~35 in~\cite{DBLP:journals/corr/abs-0903-5505}.

\begin{figure}[thb!]
  \centering
  \includegraphics[width=0.9\columnwidth]{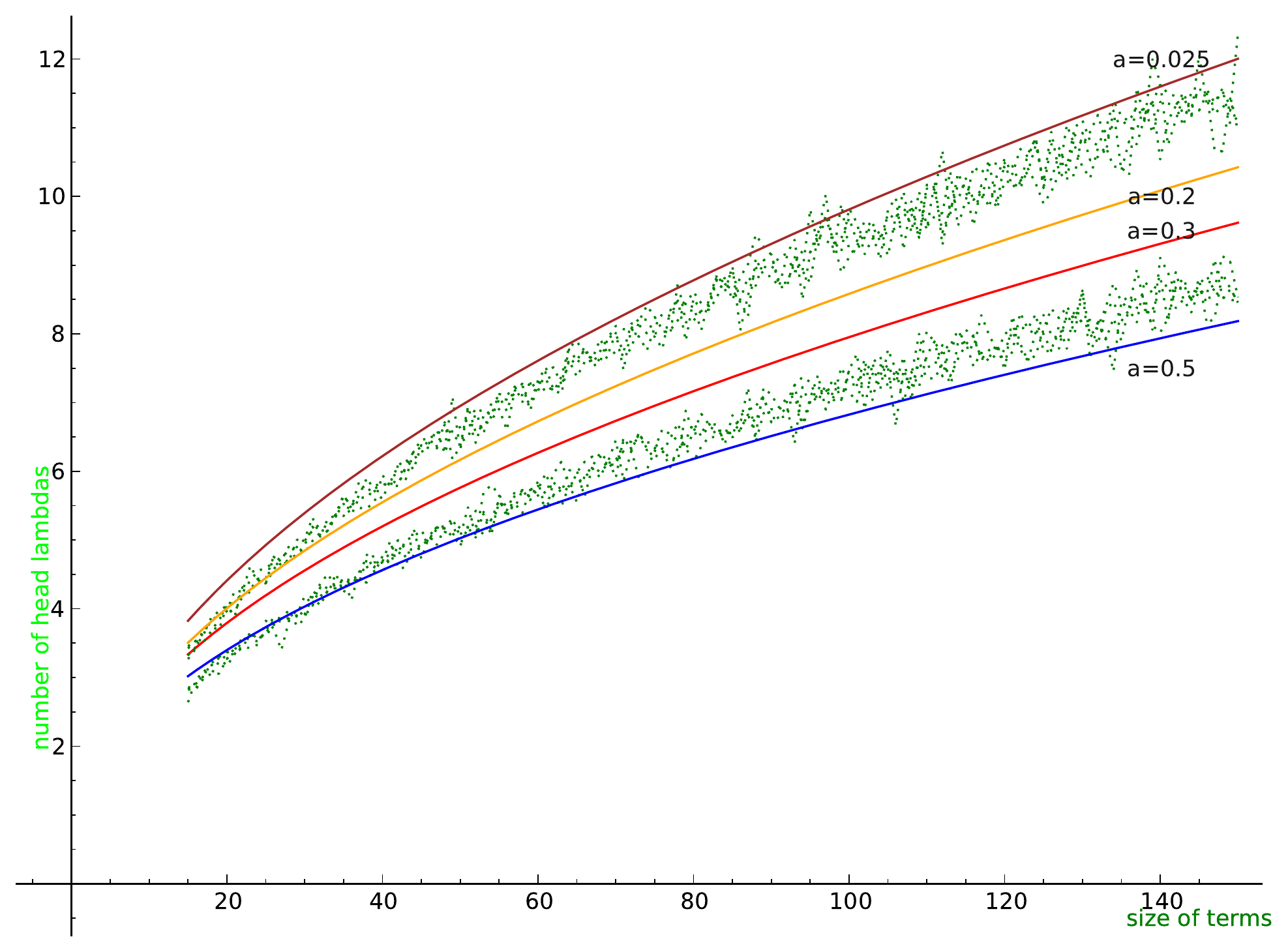}
  \caption{\textsf{Bottom}: average number of head $`l$'s per closed term.
    \textsf{Top}: average number of head $`l$'s per closed normal form.
    \textsf{In between}: curves $\sqrt{\frac{n}{\ln(n)^a}}$ for $a=0.025,0.2,0.3,0.5$.}
  \label{fig:L-in-terms}
\end{figure}

    \subsection{Ratio of simply typable terms among all terms}
    \label{sec:ratio-simply-typed}
    It is interesting to investigate the ratio of simply typable closed terms among
    all closed terms.
    There are $851\,368\,766$ closed $`l$-terms of size $11$, whereas
    there are $63\,782\,411$ closed $`l$-terms of size~$10$. Therefore, we performed
    computations for closed terms of size less than $11$. In fact, one cannot go much
    further due to the superexponential growth of the sequence enumerating closed
    terms.
    Table~\ref{tab:array9} gives the ratio of simply typable closed terms over all closed terms by an
    exhaustive examination of the closed terms up to $10$.
 \begin{table*}[htb!]
   \centering
    \begin{displaymath}
      \begin{array}{  c @{\qquad}   c   c   c   c c c c }
          \textbf{size} & 4 & 5 & 6 & 7 & 8 & 9 &10\\
          \hline
          \textbf{nb of terms} & 82 & 579 & 4\,741 & 43\,977 & 454\,283 & 5\,159\,441
          & 63\,782\,411\\
          \textbf{nb of typables} & 40& 238 & 1\,564 & 11\,807 & 98\,529 &
          904\,318& 9\,006\,364\\
          \textbf{ratio} & 0.4878 & 0.4110 & 0.3299  & 0.2684 & 0.2168& 0.1752& 0.1412\\
        \end{array}
      \end{displaymath}
   \caption{Numbers and ratios of simply typable closed terms up to size $10$}
   \label{tab:array9}
 \end{table*}
 For closed terms of size $8$ or larger, we computed the ratio by the Monte Carlo
 method. The results are given in Table~\ref{tab:array8}. We added the sequence of the numbers
 of simply typable closed terms of a given size to the \emph{On-line
   Encyclopedia of Integer Sequences} and it can be found under the number \textbf{A220471}.
 \begin{table*}[htb!]
   \centering
     \begin{scriptsize}
       \begin{math}
         \begin{array}{  c  @{\quad} r   r   r   r   r   r   r   r   r   r   r   r   r   r  r }
           \hline
           \textbf{size} & 8 & 9 & 10 & 11 & 12 & 13 & 14 & 15 & 16 & 20 & 30 & 40 & 45 & 50 \\
           \hline
           \textbf{ratio} & .216 &.175 & .141 & .111 & .089 & .073  & .056 & .047 & .039 &
           .0014 & .0012& . 0003
           & .00005 & {\scriptscriptstyle < 10^{-5}} \\
           \hline
         \end{array}
       \end{math}
     \end{scriptsize}
   \caption{Ratios of simply typable closed terms (of size at least $8$)}
   \label{tab:array8}
 \end{table*}

 We conclude that simply typable closed terms become very scarce as the size of the
 closed terms grows, falling to less than one over $100\, 000$ when the size gets
 larger than $50$.  Likewise, we have done the same task for normal forms.  We got
 the ratio by an exhaustive examination of normal forms up to $10$ in
 Table~\ref{tab:nfUpToSeven} and by the Monte Carlo method thereafter in
 Table~\ref{tab:ratioNF}.

\begin{table}[htb!]
  \centering
  \begin{displaymath}
    \begin{array}{  c @{\qquad}   c   c   c   c  c c c}
 \textbf{size} & 4 & 5 & 6 & 7 & 8 & 9 & 10  \\
\hline
\textbf{nb of NF} & 53&323&2\,359&19\,877&188\,591&1\,981\,963 & 22\,795\,849\\
\textbf{nb of typable NF} & 23&108&618&4\,092&30\,413&252\,590 & 2\,297\,954\\
\textbf{ratio} & 0.4339&0.3343&0.2619&0.2058 & 0.1612& 0.1274 & 0.1008\\
    \end{array}
  \end{displaymath}
  \caption{Numbers and ratios of simply typable closed normal forms up to size $10$}
\label{tab:nfUpToSeven}
\end{table}
\begin{table*}[htb!]
  \centering
  \begin{math}
    \begin{scriptsize}
      \begin{array}{  c @{\quad} c   c   c   c   c   c   c   c   c   c   c   c   c   c }
        \hline
        \textbf{size} & 8 & 9 & 10 & 11 & 12 & 13 & 14 & 15 & 16 & 20 & 30 & 40 & 45  \\
        \hline
        \textbf{ratio} & .159&.128  & .102 & .079& .063 &.049& .040  &.031  &.024 &.010 &.0006& 
        2.10^{-5} &{\scriptscriptstyle < 10^{-5}}\\
        \hline
      \end{array}
    \end{scriptsize}
  \end{math}

  \caption{Ratios of simply typable closed normal forms}
\label{tab:ratioNF}
\end{table*}

\subsection{Distribution of simply typable lambda terms among terms}

\begin{figure}[htb!]
  \centering
   \includegraphics[width=0.5\textwidth]{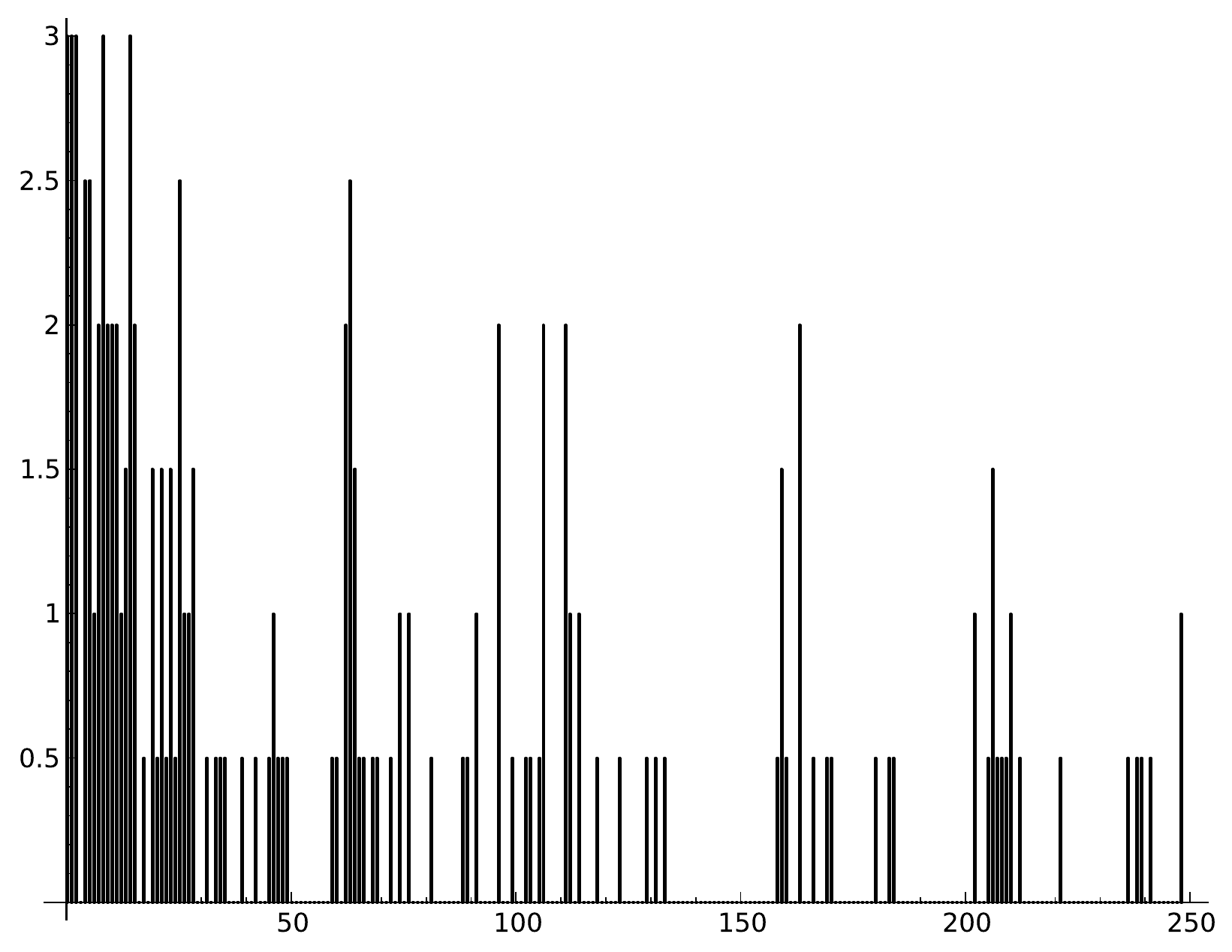}
\ifJFP
\else
\\\centerline{\xymatrix@C 20pt{&& \ar@{->}[ll]_{abstractions} &&&&
    \ar@{->}[rr]^{applications}&&}}
\fi
  \caption{Distribution of simply typable closed $`l$-terms of size \textbf{$25$}. $250$ segments on the
    horizontal axis, percentage ($0\%$ -- $3\%$) of typable closed $`l$-terms in segments on the vertical axis.}
  \label{fig:dist_typed_size_25}
\end{figure}

\begin{figure}[htb!]
  \centering
   \includegraphics[width=0.5\textwidth]{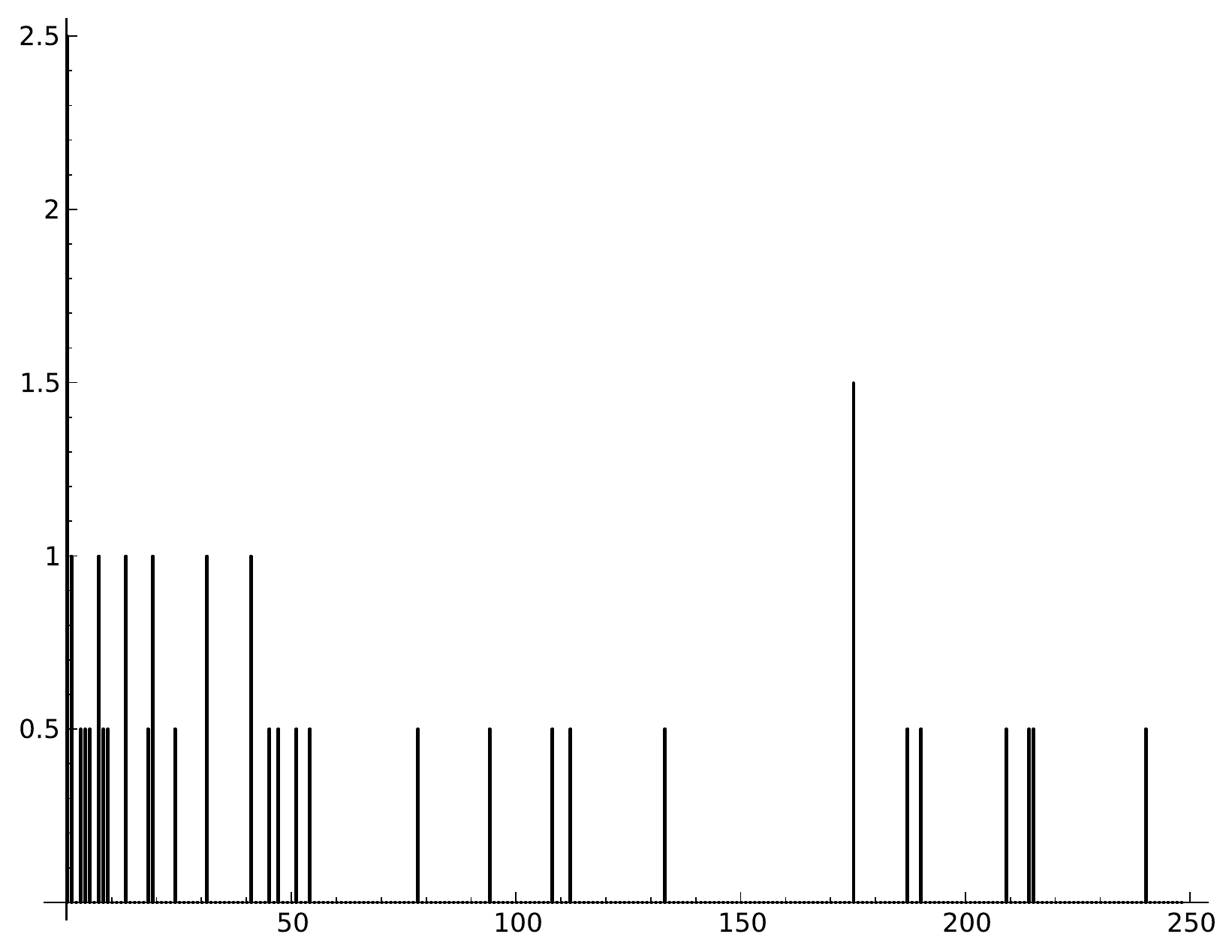}
 \ifJFP
\else
\\\centerline{\xymatrix@C 20pt{&& \ar@{->}[ll]_{abstractions} &&&&
    \ar@{->}[rr]^{applications}&&}}
\fi
\caption{Distribution of simply typable closed $`l$-terms of size
  \textbf{$30$}. $250$ segments on the horizontal axis, percentage ($0\%$ -- $2.5\%$)
  of typable closed $`l$-terms in segments on the vertical axis.}
  \label{fig:dist_typed_size_30}
\end{figure}

We said that simply typable terms are scarce, but we may wonder what scarce exactly
means.  More precisely, we may wonder how terms are distributed.  To provide an
answer to this question, we conducted experiments to approximate the distribution of
typable closed $`l$-terms  in segments of the interval $[1..T_{0,n}]$. We divided the
interval into regular segments and computed the ratio of simply typable terms
for a random sample of terms in each segment. 
Figure~\ref{fig:dist_typed_size_25} is typical of the results we got.  This
corresponds to an experiment on closed terms of size $25$ on $250$ segments with tests for
simple typability on $200$ random closed terms in each segment.  For each segment the height
of the vertical bar represents the ratio of typable closed terms to general closed terms in the corresponding
segment.   The simply typable closed terms are not uniformly distributed.  They
are more concentrated on the left of the interval corresponding to closed terms with low
numbers.  Those closed terms correspond to closed terms starting more often with abstractions than
with applications and this is recursively so for subterms giving the impression of
rolling waves.  For instance, there are $2\%$ to $3\%$ of typable closed terms (of size
$25$) starting with many abstractions, whereas for closed terms starting with many
applications there are large subintervals with almost no typable closed terms.
Figure~\ref{fig:dist_typed_size_30}, which gives the same statistics for closed terms of
size $30$, shows that typable closed terms get more scarce as the size of the closed terms grows.

The typable closed normal forms are even more scarcely distributed. As a comparison,
we drew the same graphs for closed normal forms (size of the closed normal forms:
$25$ and $30$, number of segments $250$, tests on $200$ closed terms) in
Figure~\ref{fig:dist_typed_NF_size_25}.  The typable closed normal forms aggregate
more on the left of the interval where closed terms start mostly with abstractions,
with peaks of $4\%$ to $6\%$ by segments.  Figure~\ref{fig:dist_typed_NF_size_30}
shows that scarcity of typable normal forms increases as the size of closed terms grows.

\begin{figure}[htb!]
  \centering
   \includegraphics[width=0.5\textwidth]{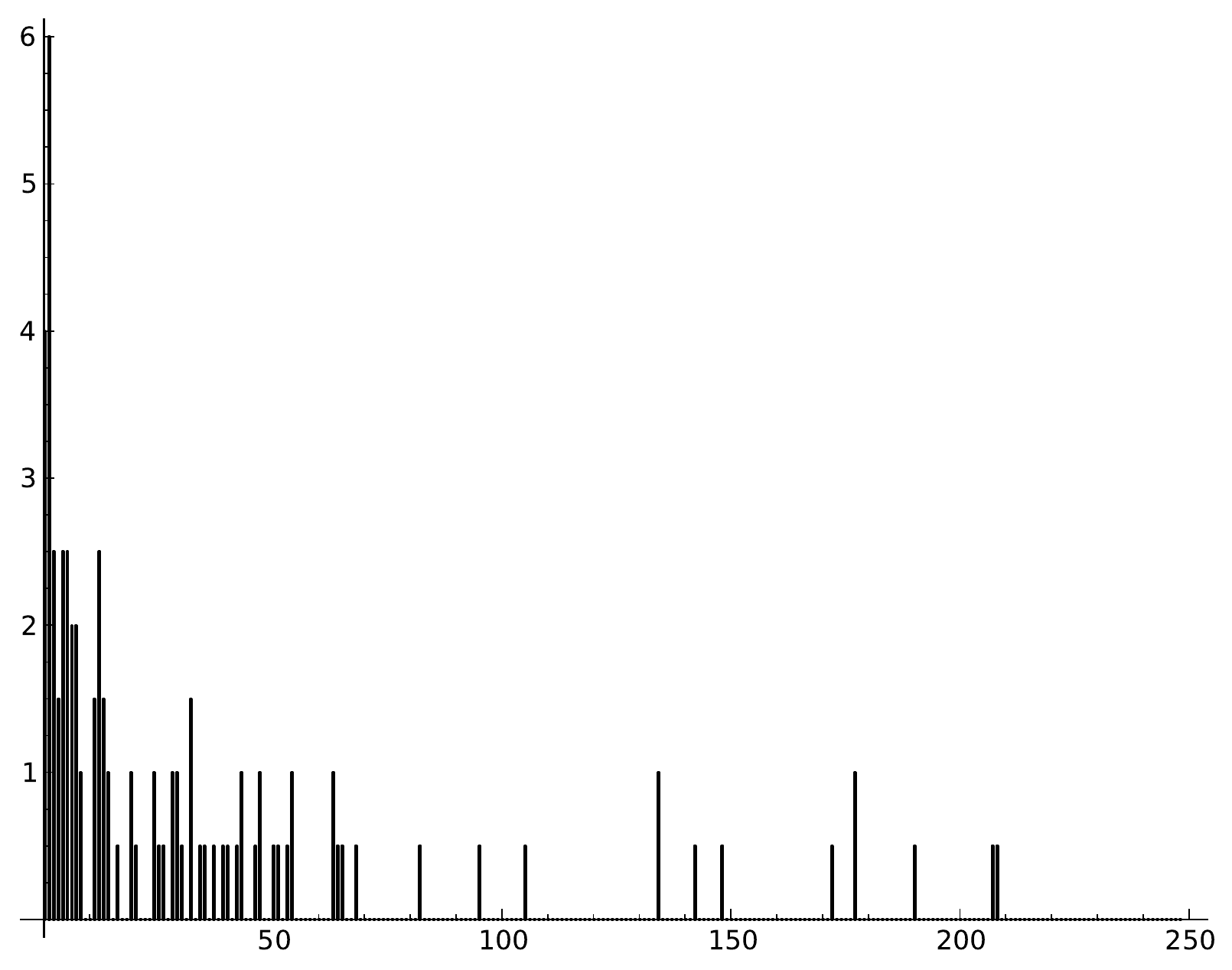}
  \ifJFP
\else
\\\centerline{\xymatrix@C 20pt{&& \ar@{->}[ll]_{abstractions} &&&&
    \ar@{->}[rr]^{applications}&&}}
\fi
  \caption{Distribution of simply typable closed normal forms  of size \textbf{$25$}. $250$ segments on the
    horizontal axis, percentage ($0\%$ -- $6\%$) of typable closed normal forms in segments on the vertical axis.}
  \label{fig:dist_typed_NF_size_25}
\end{figure}

\begin{figure}[htb!]
  \centering
   \includegraphics[width=0.5\textwidth]{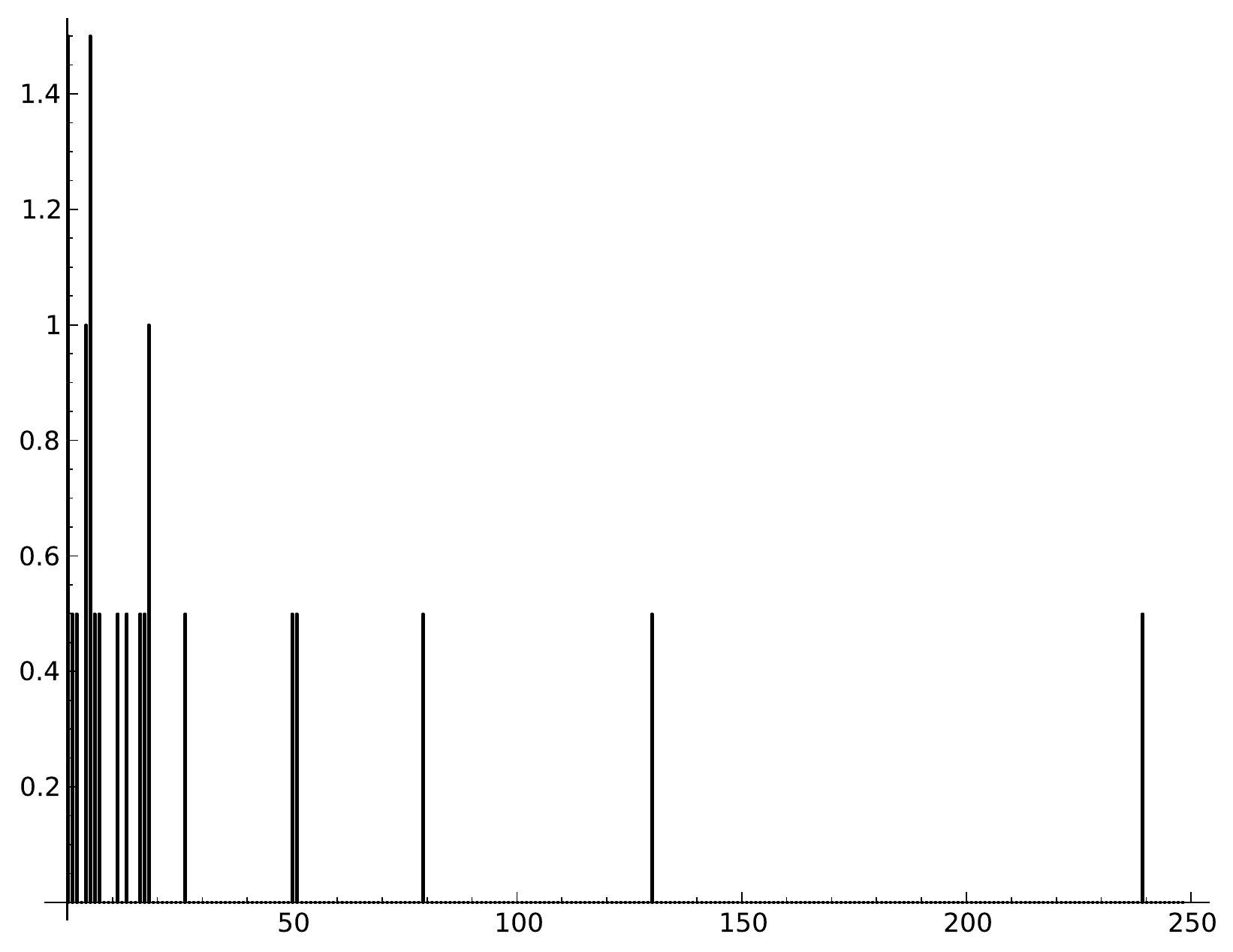}
  \ifJFP
\else
\\\centerline{\xymatrix@C 20pt{&& \ar@{->}[ll]_{abstractions} &&&&
    \ar@{->}[rr]^{applications}&&}}
\fi
  \caption{Distribution of simply typable closed normal forms  of size \textbf{$30$}. $250$ segments on the
    horizontal axis, percentage ($0\%$ -- $1.45\%$) of typable closed normal forms in segments on the vertical axis.}
  \label{fig:dist_typed_NF_size_30}
\end{figure}

%

\section{Related work}
\label{sec:related-work}

There are very few papers on counting $`l$-terms, whereas counting first order
terms is a classical domain of combinatorics. Apparently, the first traces of
counting expressions with (unbound) variables can be attributed to Hipparchus of
Rhodes (c. 190--120~BC) (see \cite{flajolet08:_analy_combin} p.~68).  Flajolet and
Sedgewick's book \cite{flajolet08:_analy_combin} is the reference on this subject.
Concerning counting $`l$-terms, we can cite only five works.  \ifJFP
\else \cite{DBLP:journals/corr/abs-0903-5505} and
\cite{gittenberger-2011-ltbuh} \fi study asymptotic behavior of formulas on counting
$`l$-terms.  Strictly speaking, they do not exhibit a recurrence formula for
counting.  In particular, David \emph{et al.} \ifJFP
\else
\cite{DBLP:journals/corr/abs-0903-5505} \fi provide only upper and lower bounds for
the numbers of $`l$-terms in order to get information about the distribution of
families of terms. For instance, they prove that ``asymptotically almost all
$`l$-terms are strongly normalizing''.  In \cite{DBLP:journals/tcs/Lescanne13} the
second author of the present paper proposes formulas for counting $`l$-terms in the
case of variables of size $1$, with more complex formulas and less results.  On
another hand, Christophe Raffalli proposed a formula for counting closed $`l$-terms,
which he derives from the formula for counting $`l$-terms with exactly $m$ distinct
free variables.  His formula appears in the \emph{On-line Encyclopedia of Integer
  Sequences} under the number \textbf{A135501}.  He considers size $1$ for the
variables.  Beside those works, John Tromp \ifJFP
\else ~\cite{DBLP:conf/dagstuhl/Tromp06} \fi
proposes a rather different way of counting $`l$-terms which deserves to be
investigated further from the viewpoint of combinatorics.  His size function works on
terms with de Bruijn indices like ours and is (in our convention of starting at $1$)
as follows:
  \begin{eqnarray*}
    |\Var{n}| &=& n +1\\
    |`l M | &=& |M| + 2\\
    |M\,N| &=& |M| + |N|+2
  \end{eqnarray*}
  producing sequence \textbf{A114852} (and sequence \textbf{A195691} for closed
  normal forms) in the \emph{On-line Encyclopedia of Integer Sequences}.  This work
  is connected to program size complexity and \emph{Algorithmic Information
    Theory}~\cite{261084}.

  As concerns random generation, Wang in
  \cite{Wang05generatingrandom,wang04:_effic_gener_random_progr_their_applic}
  proposed algorithms for random generation of untyped $`l$-terms in the spirit of the
  counting formula of Raffalli for~\cite{Wang05generatingrandom} and in the spirit of
  $T_{n,m}$ for~\cite{wang04:_effic_gener_random_progr_their_applic}.  On term
  generation, we can also mention two works. In \cite{DBLP:conf/haskell/DuregardJW12}
  the authors enumerate and generate many more structures than
  $`l$-terms. In~\cite{DBLP:journals/jfp/KennedyV12}, the authors address a problem
  similar to ours. Indeed they play on the duality \emph{encoder-decoder} or
  \emph{ask-build}, when we speak of \emph{counting-generating} or
  \emph{ranking-unranking}.  Our unranking program (Figure~\ref{fig:prog-gen}) can be
  made easily a game, with questions like ``Is ${k \le (t (n-1) (m+1))}$?'' or like
  ``Is ${k > (t (n-1) (m+1))}$?'', but Kennedy and Vytiniotis do not know the precise
  range of their questions since they do not base their generation on counting. Since
  the size is not a parameter, their games may have unsuccessful issues and their
  programs can raise errors and are only error-free on well-formed games.  Pa{\l}ka
  \cite{palka12:_testin_compil,Palka:2011:TOC:1982595.1982615} uses generation of
  typable $`l$-terms to test Haskell compilers.  He acknowledges that, due to his
  method, he cannot guarantee the uniformity of his generator (see discussion in
  \cite{palka12:_testin_compil} p. 21 and p.~45).  Nonetheless, he found eight
  failures and four bugs in the \emph{Glasgow Haskell Compiler} demonstrating the
  interest in the method, probably due the ability of generating large terms. \ifJFP
\else {Rodriguez~Yakushev} \& Jeuring
  \cite{DBLP:conf/aaip/YakushevJ09} \fi study the feasibility of generic
  programming for the enumeration of typed terms. The given examples are of size $4$
  or $5$, no realistic examples are provided, randomness is not addressed and the
  authors confess that their algorithm is not efficient.  Knowing that there are
 $ 9\,006\,364$ simply typable closed terms of size $10$, one wonders if there is an
  actual use for such enumeration and it seems unrealistic to utilize enumeration for
  larger numbers.  The ``related work'' section of \cite{DBLP:conf/aaip/YakushevJ09}
  covers similar approaches, which all consist in cutting branches.  For this reason
  they do not generate terms uniformly.  A presentation of tree-like structure
  generation and a history of combinatorial generation is given \ifJFP
\else in \cite{KnuthTAOCP_4_4}\fi.

  Since we cited, as an application, the random generation of terms for the
  construction of samples for debugging functional programming compilers and the
  connection with languages with bound variables, it is sensible to mention
  \textsf{Csmith}~\cite{DBLP:conf/pldi/YangCER11}, which is the most recent and the
  most efficient bug tracker of C compilers.  It is based on random program
  generation and uses filters for generating programs enforcing semantic
  restrictions, like ours when generating simply typable terms.  However, the generation
  is not based on unranking, therefore \textsf{Csmith} lacks the ability to construct
  test case of a specific size on demand, but \textsf{Csmith} can generate large
  terms, which reveals to be useful, since the greatest number of distinct crash
  errors is found by programs containing 8K-16K tokens. However, one may wonder if
  this feature is not a consequence of the non-uniformity of the distribution.

\section{Acknowledgments}
\label{sec:acknowledgments}

Clearly Nikolaas de Bruijn and  Philippe Flajolet were influential all along this research.  We would like to
thank Marek Zaionc for stimulating discussions and for setting the problem of
counting $`l$-terms, Bruno Salvy for his help in the proof of
Proposition~\ref{prop:equiv}, Olivier Bodini, Jonas Dureg{\aa}rd, Dani\`{e}le Gardy,
Bernhard Gittenberger, Patrik Jansson, Jakub Kozik, John Tromp and the referees of this paper for
their useful suggestions and interactions.

\section{Conclusion}
\label{sec:conclusion}

This paper opens tracks of research in two directions, which are intrinsically
complementary, namely counting and generating, aka ranking and unranking. On counting terms, some hard problems
remain to be solved.  Probably the hardest and the most informative one is to
give an asymptotic estimation for the numbers of closed terms of size $n$. It seems that big
obstacles remain to be hurdled before getting a solution, since combinatorial
structures with binders have not been studied so far by combinatorists.  On generation of
terms, implementations have to be improved to go further in the production of
uniformly distributed terms, in particular, of uniformly generated typable terms.  


\appendix

\section{Terms with exactly $m$ distinct free variables}
\label{sec:number-terms-with}

Here we study the numbers of terms with exactly $m$ distinct
free variables, the formulas for counting those numbers and their relations with
quantities we considered.

\subsection{A formula}
\label{sec:formula}

Let us show how to derive the formula for counting $`l$-terms with exactly $m$
distinct free variables.  This formula is adapted from a similar one when variables
have size $1$ due to Raffalli (\emph{On-line Encyclopedia of Integer Sequences}
under the number \textbf{A135501}).  We assume that terms are built with usual variables
(not de Bruijn indices) and that they are equivalent up to a renaming of bound variables
and up to $`a$-conversion.  Let us denote the number of $`l$-terms of size $n$ with
exactly $m$ distinct free variables by $f_{n,m}$.

Notice first that there is no term of size $0$ with no free variable, hence $f_{0,0}
= 0$.  There is one term of size $0$ with one free variable, hence $f_{0,1} = 1$.  The maximum number of variables for
a $`l$-term of size $n$ is when the only operators are applications and all the
variables are different.  One has then a binary tree with $n$ internal nodes and
$n+1$ leaves holding $n+1$ variables.  This means that for $m$ beyond $n+1$ variables
there is no term of size $n$ with exactly $m$ distinct free variables.  Hence
\[ f_{n,m} =0 \qquad \textrm{when} ~ m > n+1.\]

In the general case, a term of size $n+1$ with $m$ free variables starts either with
an abstraction or with an application. Terms starting with an abstraction, say $`l
x$, on a term $M$ contribute in two ways: either $M$ does not contain $x$ as a free
variables or $M$ contains $x$ as a free variable.  There are $f_{n,m}$ such $M$'s in
the first case and $f_{n,m+1}$ in the second.  This gives the two first summands
$f_{n,m}+f_{n,m+1}$ in the formula.  Now, let us see how terms starting with an
application look like.  Assume they are of the form $P\,Q$ and of size $n+1$.  For
some $p \leq n$, the term $P$ is of size $p$ and $Q$ is of size $n-p$.  These
terms share $c$ common variables ($0\le c \le m$), while $P\, Q$ has $m$ distinct
free variables altogether.  The term~$P$ has $k$ distinct free variables, which do not occur in $Q$,
hence $P$ has $k+c$ distinct free variables altogether.  The term $Q$ has $m-k$
distinct free variables.  Therefore, given a set of private variables for $P$, a set of
common variables, and a set of private variables for $Q$, there are $f_{p,k+c} f_{n-p,
  m-k}$ possible pairs $(P,Q)$.  There are ${m \choose c}$ ways to choose the $c$
common variables among $m$ and there are ${m-c \choose k}$ ways to split the
remaining variables into $P$ and $Q$, namely $k$ for $P$ and $m-c-k$ for $Q$, hence
the third summand of the formula:
\[\sum_{p=0}^{n} \sum_{c=0}^{m} \sum_{k=0}^{m - c} {m \choose c} {m - c\choose k} f_{p,k+c} f_{n-p,m-k}.\]
Now, we obtain the whole formula:
\begin{displaymath}
f_{n+1,m} \ =\ f_{n,m} + f_{n,m+1} + \sum_{p=0}^{n} \sum_{c=0}^{m} \sum_{k=0}^{m - c} {m \choose c} {m - c\choose k} f_{p,k+c} f_{n-p,m-k}.
\end{displaymath}

\subsection{Relations between $T_{n,m}$ and $f_{n,m}$}
\label{sec:relat-betw}

The number of terms of size $n$ with exactly $i$ indices in $[\Var{1}..\Var{m}]$ is
\({m \choose i} f_{n,i} .\)
Therefore the number of terms with indices in $[\Var{1}..\Var{m}]$ is:
\[T_{n,m} \quad = \quad \sum_{i=0}^{m} {m \choose i} f_{n,i}.\]
By the inversion formula (\cite{GraKnuPat} p.~192), we get:
\[f_{n,m} \quad = \quad \sum_{i=0}^m (-1)^{m+i} {m \choose i} T_{n,i}.\] This shows
with no surprise that $f_{n,m}$ and $T_{n,m}$ are simply connected.  Knowing
that the $T_{n,m}$'s can be easily computed, this provides a formula simpler than
Raffalli's to compute the $f_{n,m}$'s.

\subsection{A relation between $f_{n,m}$ and $c_{n,i}$}
We write $R^{(m)}_i$ the number of surjections from $[\Var{1}..\Var{i}]$ to $[\Var{1}..\Var{m}]$.
To get a relation between $f_{n,m}$ and $c_{n,i}$, we can reproduce the process with
which we associated $T_{n,m}$ and $c_{n,i}$ (Section~\ref{sec:an-comb-interpr}), but
instead of applications from $[\Var{1}..\Var{i}]$ to $[\Var{1}..\Var{m}]$, we have surjections from $[\Var{1}..\Var{i}]$
to $[\Var{1}..\Var{m}]$, since this time we count terms with exactly $m$ variables and all the de
Bruijn indices must be reached by the applications.  Therefore
\begin{eqnarray*}
  f_{n,m} &=& \sum_{i=0}^n c_{n,i} R^{(m)}_i.
\end{eqnarray*}
Recall that
\[R^{(m)}_i \quad = \quad \sum_{j=0}^m {m \choose j} (-1)^j (i-j)^m.\]
We can now go further in the expression of $f_{n,m}$.
\begin{eqnarray*}
   f_{n,m} &=& \sum_{i=0}^n c_{n,i} \sum_{j=0}^m {m\choose j} (-1)^j (m-j)^i\\
   &=&\sum_{i=0}^n c_{n,i} \sum_{k=0}^m {m\choose k} (-1)^{m-k} k^i\\
   &=& \sum_{k=0}^m  {m\choose k} (-1)^{m-k} \sum_{i=0}^n c_{n,i} k^i\\
   &=& \sum_{k=0}^m   (-1)^{m+k} {m\choose k}\ T_{n,k}.
\end{eqnarray*}
which is another proof of the formula of Section~\ref{sec:relat-betw}.

\section{A program for testing simple typability}
\label{sec:typ_prog}

In this section we give a simple \textsf{Haskell} program for \emph{testing simple
typability} of a term also called \emph{type reconstruction}.  The program which
works on the types \textsf{Type} and \textsf{Equation}:
\begin{verbatim}
data Type = Var Int
          | Arrow Type Type

type Equation = (Type,Type)            
\end{verbatim}
has three parts. First, a function builds the set of typability equational
constraints of a closed term. This function called \textsf{buildConstraint}
(Figure~\ref{fig:bldcons}) takes a term and returns its potential principal type,
which will be made explicit after solving the constraints, and a list of equational
constraints.  It requires a function \textsf{build} which will be called through the
terms.  Along its traversal of the term, the function \textsf{build} has to know the
depth $d$ (the number of $`l$'s it crossed). Moreover, \textsf{build} creates type
variables.  Actually, a constraint builder creates type variables in two situations:
when it creates a context for the first time, that is when it deals with a de Bruijn
index, and when it creates the type to be returned by an application.  Since type
variables are objects of the form \textsf{Var~i}, where \textsf{i} is an
\textsf{Int}, \textsf{build} takes an \textsf{Int} which is increased whenever a new
type variable is created. We call the latter a \emph{cursor} and denote it by
\textsf{cu}. \textsf{build} returns a 4-uple, namely the potential principal type of
the term, a context (a list of types associated with de Bruijn indices), a set of
equational constraints and the updated cursor.

\begin{figure}[!btph]
  \centering
  \begin{normalsize}
\begin{verbatim}
buildConstraint :: Term -> (Type, [Equation])
buildConstraint t = 
  let (ty,[],constraint,_) = build t 0 0 
  in (ty, constraint)
  where
    build :: Term -> Int -> Int -> (Type, [Type], [Equation],Int)
    build (Index i) d cu = 
      let ii = fromIntegral i 
      in (Var (cu+ii-1), [Var j | j<-[cu..cu+d-1]],[],cu+d)
    build (Abs t) d cu = 
      let (ty,(a:cntxt),constraint,cu') = build t (d+1) cu 
      in ((Arrow a ty),cntxt,constraint,cu')
    build (App t1 t2) d cu = 
      let (ty1, cntxt1, constraint1, cu1) = build t1 d cu 
          (ty2, cntxt2, constraint2, cu2) = build t2 d cu1
          ty = (Var cu2) in (ty,
                             cntxt1,
                             (ty1,(Arrow ty2 ty)):(zip cntxt1 cntxt2) 
                                ++ constraint1 ++ constraint2, 
                             cu2+1)
\end{verbatim}
  \end{normalsize}
  \caption{The function \textsf{buildConstraint}}
\label{fig:bldcons}
\end{figure}
To solve equational constraints we use a method based on transformation
rules~\cite{GallierSnyderTCS89,JouannaudKirchner-rob91}. For that, we use a function
\textsf{decompose} which splits an equation when both sides are arrow types.
Moreover, when \textsf{decompose} meets an equation $`s_1 "->" `s_2 = `a$, the latter
is transformed into $`a = `s_1 "->" `s_2$.

\begin{verbatim}
decompose :: Equation -> [Equation]
decompose ((Arrow ty1 ty2), (Arrow ty1' ty2')) = 
  decompose (ty1,ty1') ++ decompose (ty2,ty2')
decompose ((Arrow ty1 ty2),(Var i)) = [(Var i,(Arrow ty1 ty2))]
decompose (ty1,ty2) = [(ty1,ty2)]
\end{verbatim}

A predicate \textsf{nonTrivialEq} is necessary to filter out the trivial
equations, i.e., of the form $`a=`a$.
\begin{verbatim}
nonTrivialEq :: Equation -> Bool
nonTrivialEq (Var i, Var j) = i /= j
nonTrivialEq (ty1, ty2) = True
\end{verbatim}
A predicate {\EUR} checks whether a given variable belongs to a composed type. This
is necessary to detect cycles. For instance, $`a = `b "->" `a$ is a cycle and shall
be detected, whereas $`a = `a$ is a trivial equation, not a cycle, and shall be
removed.
\begin{alltt} 
(\EUR) :: Type -> Type -> Bool 
(Var i) \EUR (Var j) = False -- strict occurrence only
(Var i) \EUR (Arrow ty1 ty2) =  (Var i) \EUR= ty1 || (Var i) \EUR= ty2
    where (Var i) \EUR= (Var j) = i == j
          (Var i) \EUR= (Arrow ty1 ty2) = (Var i) \EUR= ty1 || 
                                       (Var i) \EUR= ty2
\end{alltt}
Once this test is done, one can replace a variable $`a$ occurring in an equation of
the form $`a = `s$ by $`s$ everywhere else in the set of equational constraints
before putting the equation $`a=`s$ in the solved part.
\pagebreak[4]
\begin{alltt}
(\(\leftarrow\)) :: Type -> Equation -> Type
(Var j) \(\leftarrow\) (Var i, ty) = if i == j then ty else Var j
(Arrow ty1 ty2) \(\leftarrow\) (Var i, ty) = 
  Arrow (ty1 \(\leftarrow\) (Var i, ty)) (ty2 \(\leftarrow\) (Var i, ty))

replace::  Equation -> Equation -> Equation
replace (Var i,ty) (ty1,ty2) = (ty1 \(\leftarrow\) (Var i,ty),  ty2 \(\leftarrow\) (Var i,ty))
\end{alltt}
The function \textsf{solve} solves the set of equational constraints.  It returns three things:
first a~list of equations which are the equations yet to be solved, second a list of
equations which are the already solved equations and third a condition that says if the set
of equational constraints is solvable or not.  In other words, \textsf{solve} returns
\textsf{True} as the third component if the set of equational constraints has a solution, that
is when the set of equational constraints is empty.  It returns \textsf{False} when it detects a
cycle. Otherwise it tries to apply the transformations whenever it is possible, that
is when the set of equational constraints is not empty. Indeed if there is no cycle and if the
set of equational constraints is not empty, a~transformation is always applicable.
\begin{normalsize}
\begin{alltt}
solve :: [Equation] -> [Equation]
                    -> ([Equation],[Equation],Bool)
solve ((Var i,ty):l) sol = 
  if Var i \EUR ty 
  then ((Var i,ty):l,sol,False) -- cycle detected
  else solve (map (replace (Var i,ty)) l) ((Var i,ty):sol)
solve (eq:l) sol = solve (filter nonTrivialEq (decompose eq) ++ l) sol
solve [] sol = ([],sol,True)
\end{alltt}
\end{normalsize}
Since we have all the ingredients, the test of typability consists in building the equational constraint and trying to
solve it. 
\begin{verbatim}
isTypable :: Term -> Bool
isTypable t = let (_,c) = buildConstraint t 
                  (_,_,b) = solve c []
              in b
\end{verbatim}
Notice that we have everything to build the principal type of the term if it is
typable. 
\end{document}